\documentclass[11pt,letterpaper]{article}

\usepackage[margin=1in]{geometry}
\usepackage{lmodern}
\usepackage[T1]{fontenc}
\usepackage[utf8]{inputenc}
\usepackage[tbtags]{amsmath}
\allowdisplaybreaks
\usepackage{amssymb,amsthm,mathtools, bbm}
\usepackage{physics}
\usepackage{hyperref}
\usepackage[svgnames]{xcolor}
\usepackage[capitalise,nameinlink]{cleveref}
\definecolor{links}{RGB}{11, 85, 255}
\definecolor{cites}{RGB}{0, 200, 0}
\definecolor{urls}{RGB}{255, 116, 0}
\hypersetup{colorlinks={true},linkcolor={links},citecolor=[named]{cites},urlcolor={urls}}
\usepackage[numbers,compress]{natbib}
\usepackage{doi}
\usepackage{tikz} 
\usepackage{pgfplots}
\pgfplotsset{compat=1.14}
\usepgfplotslibrary{fillbetween}
\usepackage{hyphenat}
\usepackage[font=footnotesize]{caption}
\usepackage{subcaption}
\usepackage{appendix}
\usepackage{nicefrac}
\usepackage{tcolorbox}

\usepackage{todonotes}

\usepackage{thm-restate}

\usepackage{algorithm}
\usepackage[noend]{algpseudocode}

\newcommand{\cI}{\mathcal{I}}

\newcommand{\R}{\mathbb{R}}
\newcommand{\insp}{v}
\newcommand{\cyclic}{\text{cyclic}}

\newcommand{\agent}{\text{a}}
\newcommand{\principal}{\text{p}}

\newcommand{\ind}[1]{\mathbbm{1}\left[#1\right]}

\newcommand{\bb}{\textbf{b}}

\newcommand{\inst}{\mathcal{I}}
\newcommand{\eqdef}{\overset{\mathrm{def}}{=\mathrel{\mkern-3mu}=}}

\theoremstyle{theorem}

\newtheorem{theorem}{Theorem}[section]
\newtheorem*{theorem*}{Theorem}
\newtheorem{lemma}[theorem]{Lemma}
\newtheorem{proposition}[theorem]{Proposition}

\newtheorem{remark}[theorem]{Remark}
\newtheorem{observation}[theorem]{Observation}

\newtheorem{definition}[theorem]{Definition}

\newtheorem{example}[theorem]{Example}

\begin{document}

\title{Contract Design Beyond Hidden-Actions}
\date{}

\author{Tomer Ezra\thanks{Harvard University, USA (tomer@cmsa.fas.harvard.edu).} 
\and Stefano Leonardi\thanks{Sapienza University of Rome, Italy (leonardi@diag.uniroma1.it, mrusso@diag.uniroma1.it).}
\and Matteo Russo\footnotemark[2]
}

\maketitle

\thispagestyle{empty}

\begin{abstract}
In the classical principal-agent hidden-action contract model, a principal delegates the execution of a costly task to an agent. In order to complete the task, the agent chooses an action from a set of actions, where each potential action is associated with a cost and a success probability to accomplish the task. To incentivize the agent to exert effort, the principal can commit to a contract, which is the amount of payment based on the task's success but not on the hidden-action chosen by the agent.

In this work, we study the contract design framework under binary outcomes where we relax the hidden-action assumption.
We introduce new models where the principal is allowed to inspect subsets of actions at some cost that depends on the inspected subset. If the principal discovers that the agent did not select the agreed-upon action through the inspection, the principal can withhold payment. 
This relaxation of the model introduces a broader strategy space for the principal, who now faces a tradeoff between positive incentives (increasing payment) and negative incentives (increasing inspection).

We devise algorithms for finding the best deterministic and randomized incentive-compatible inspection schemes for various assumptions on the inspection cost function. In particular, we show the tractability of the case of submodular inspection cost functions.   We complement our results by showing that it is impossible to efficiently find the optimal randomized inspection scheme for the more general case of XOS inspection cost functions, and that there is no PTAS for the case of subadditive inspection cost functions.

\end{abstract}
\clearpage
\pagenumbering{arabic}

\section{Introduction}

Contract theory is the field of economics that studies the interaction between two interested parties, a principal and an agent, and seeks to answer the foundational question: ``How do we incentivize people to work?'' The original model that captures the tension between the two rational individuals  involved is called \emph{hidden-action} principal-agent interaction \cite{HolmstromM87, HolmstromM91, GrossmanH83}. The principal aims to incentivize an agent to take some costly action among a set of $n$ potential actions $A$. Each action $a\in A$ that the agent can take incurs a cost $c(a)$ for the agent and results in a distribution over the set of $m$ potential outcomes $O=\{o_1,\ldots,o_m\}$, where each outcome $o_j$ is associated with a reward for the principal $r_j$.  
In order to incentivize the agent to exert effort, that will more likely lead to a favorable outcome, the principal can post a payment scheme (or a contract) that maps the set of outcomes to payments for the agent. The goal of the principal is to design a contract that incentivizes the agent to exert effort, which will maximize her own expected utility (reward minus payment). The postulates defining contract design are \textit{limited-liability} which states that the payments are only one-way from the principal to the agent, and \textit{hidden-action}, which states that the payments can only depend on the realized outcome and not the action taken by the agent. 
When outcomes are binary ($m=2$), it is without loss of generality to assume that $r_1=0,r_2=1$, i.e., success or failure, and each action $a$ is associated with the probability of success $f(a)$. 
Then, the goal of the principal is to find a single-parameter $\alpha \in [0,1]$ (the payment in case of success), such that the agent is incentivized to take some action $a^*$ (i.e., $a^* \in \arg\max_{a} \alpha f(a)-c(a)$) with maximum utility for the principal (i.e., maximizing $(1-\alpha)f(a^*)$). 

However, in many real-life scenarios,  the principal may be able to examine specific actions performed by the agent. 
Consider for example a bank (the principal) that partners with a loan officer (the agent) to evaluate and approve loan applications. The loan officer can choose between conducting thorough evaluations, which are time-consuming and costly, or performing quick but superficial reviews, increasing the risk of approving high-risk loans. Since the bank cannot directly observe the evaluation process, it designs a contract that combines per-loan commissions with random inspections of approved loans. There are many potential types of inspections the bank can do, and each combination of them is associated with (potentially combinatorial) cost. If audits reveal insufficient evaluations, penalties or reduced payments are imposed, incentivizing the loan officer to align his actions with the bank’s low-risk policies while balancing inspection costs.
Another example is where a store owner can send inspectors to check the performance of his salespeople, and a company can install activity monitoring software on the computers of its employees. Another motivating example inspired by the concurrent work of \citet{FallahJ23}, involves a principal aiming to incentivize a worker to choose an action of the highest safety level among $L$ available levels, with the principal able to verify whether the safety level of the selected action is at least $\ell$. While \citet{FallahJ23} focused on scenarios with only two safety levels (safe and unsafe), our framework addresses more complex cases such as having multiple levels.
In such scenarios, the principal could base the payment not only on the outcome but also on the result of the inspection. 

\paragraph{Contracts with Inspection Schemes.} In this work, we relax the hidden-action postulate by allowing the principal to inspect the actions taken by the agent, incurring an inspection cost. 
We focus on natural contracts, where the principal must commit to a payment and inspection distribution in advance and is only permitted to withhold payment if the agent is caught not performing the desired action.\footnote{If the principal is allowed to condition payment on the event of an inspection, she could easily extract the entire welfare by inspecting the welfare-maximizing action $i^*=\arg\max_{i} f(i)-c(i)$ with a very low probability $\varepsilon$, and paying the agent $\frac{c(i^*)}{\varepsilon}$ only if this action is actually inspected. The agent is incentivized to select action $i^*$ since the utility from it is non-negative, and all other actions lead to receiving a payment of $0$. 
The principal's utility is then $f(i^*)-c(i^*) -\varepsilon \cdot v(i^*)$, where $v(i^*)$ is the cost of inspecting action $i^*$. When $\varepsilon\rightarrow 0$, then this utility converges to the maximum social welfare of the instance. This solution concept would be not only mathematically trivial but also not very useful in practice since it uses that the agent is willing to accept a payment just in an arbitrarily small probability and being paid in cases that are fully controlled by the principal.}

Our setting can be described as follows: We are given a set of actions $A$ the agent can carry out, each associated with a cost $c(a)$, %
a success probability $f(a)$, %
and an inspection cost function $v:2^A \rightarrow \R_{\geq 0}$.
The principal then proposes an inspection scheme described by a suggested action $i\in A$, a payment $\alpha$, and an inspection distribution $p$ over subsets of actions.
If the agent takes the suggested action $i$, then he is paid $\alpha$ in case of success.
If the agent takes some other action $j\neq i$, then, only in the case that neither action $i$ nor action $j$ are inspected, the agent is paid $\alpha$ in case of success. This is what we call \textit{fully-observable} inspection model. We say that an inspection scheme is incentive-compatible (IC) if the suggested action is the best response for the agent.
The principal's goal is to find the IC inspection scheme that maximizes the utility of the principal.

We also extend the inspection model to a setting where the inspection outcome is only \textit{partially-observable}. Namely, when the principal inspects some set $S \subseteq A$, she only receives as feedback a single bit indicating whether the action $j$ taken by the agent was in $S$  or not. This is in contrast to the \textit{fully-observable} model where, by inspecting set $S \subseteq A$, the principal knows \textit{exactly} which action $j \in S$ (if any) the agent has taken.

\medskip

Regarding the example of a manager hiring an inspector to evaluate the salespeople's performance, it may not be feasible or cost-effective to always conduct inspections. Consequently, the principal may choose to inspect the salespeople with a certain probability, thereby reducing inspection costs. This strategic approach could potentially result in increased utility, as illustrated in the following example.
\begin{tcolorbox}
\begin{example}\label{ex:gap}
    Let us consider the following instance with three actions $A = \{\bot, b, g\}$: $c(\bot)=0,~f(\bot)=\frac{1}{10},~ v(\bot)=1,~c(b)=\frac{1}{10},~f(b)=\frac{1}{2}, ~v(b)=1,~c(g)=\frac{7}{20},~f(g)=1,~v(g) = \frac{1}{10}$. We assume $v$ to be additive. 
    \begin{itemize}
        \item Under no inspection, we can observe that $\alpha \in [0, \frac{1}{4})$ incentivizes action $\bot$ and offering $\alpha = 0$ yields utility $\frac{1}{10}$; $\alpha \in [\frac{1}{4}, \frac{1}{2})$ incentivizes action $b$ and  offering $\alpha = \frac{1}{4}$ yields utility $\frac{3}{8} > \frac{1}{10}$; $\alpha \in [\frac{1}{2}, 1]$ incentivizes action $g$ and offering $\alpha = \frac{1}{2}$ yields utility $\frac{1}{2} > \frac{3}{8}$. Thus, the highest utility for the principal is $\frac{1}{2}$  and corresponds to action $g$.
        \item If the principal is only allowed to use a deterministic inspection schemes, i.e., always inspect a specific subset of actions, then inspecting any action different from action $g$ yields strictly negative utility to the principal. The principal can thus ask the agent to do action $g$, inspect it, and offer $\alpha=\frac{7}{20}$, leading to a utility of $\frac{13}{20} - \frac{1}{10} = \frac{11}{20} > \frac{1}{2}$, i.e., it is strictly larger than the highest utility under no inspection.
        \item Under randomized inspection schemes, the principal can ask the agent to take action $g$, inspect it with probability $\frac{3}{7}$ (and all other actions with probability $0$), offer $\alpha=\frac{7}{20}$ (so as to incentivize action $g$ over the other actions), which leads to a utility of $\frac{13}{20} - \frac{3}{70} = \frac{17}{28} > \frac{11}{20}$, i.e., it is larger than in the previous two scenarios.
    \end{itemize}
\end{example}
\end{tcolorbox}

In our study, we make the assumption that the inspection cost of actions falls within the complement-free hierarchy introduced by \citet{LehmannLN06}, and widely used in the context of combinatorial auctions \cite{DobzinskiNS05, AssadiKS21, CorreaC23}, combinatorial and multi-agent contracts \cite{DuttingEFK21, DuttingEFK23}, as well as Prophet Inequalities and pricing \cite{CorreaC23, FeldmanGL15} to mention a few. This hierarchy encompasses various real-world scenarios, as can be seen in the example with the company using a monitoring software, where it might be cheaper to buy a monitoring software that checks multiple measures simultaneously rather than purchasing a monitoring software for each measure separately. 
The inspection cost, in this case, exhibits subadditivity or potentially even submodularity (see formal definition in Section~\ref{sec:model}).

\medskip

Our guiding question is that of investigating the \emph{tractability of the optimal inspection scheme under different classes of cost functions}. We aim to characterize when the best deterministic or randomized inspection scheme can be efficiently found.

\subsection{Our Contribution and Techniques}
We first show how to design a simple algorithm that finds the optimal deterministic IC inspection scheme in time polynomial in the number of actions for \emph{every} monotone combinatorial cost function (Theorem~\ref{thm:det}) in the fully observable inspection model. 
To find the best deterministic inspection scheme, we make the following observations. 
Either it is beneficial to inspect the suggested action $i$ and pay the minimal amount to ensure that the action has non-negative utility $\frac{c(i)}{f(i)}$; otherwise, the payment should be such that the agent's utility from action $i$ is the same as from some other action $j$. This allows us to consider a polynomial number of values of $\alpha$ such that either $\alpha=\frac{c(i)}{f(i)}$ or $\alpha=\frac{c(i)-c(j)}{f(i)-f(j)}$ for some action $j$. For each such $\alpha$, the set of actions that need to be inspected can be calculated easily: They are the actions with strictly larger utility for the agent under contract $\alpha$. 

\medskip

We then shift our focus to randomized inspection schemes. Our first main result proves that for the class of \emph{submodular} inspection cost functions (i.e., decreasing marginals), one can find the optimal randomized inspection scheme  in polynomial time. In particular, we show:

\medskip

\noindent \textbf{Main Result 1 (Theorem~\ref{thm:submod}):} In the \emph{fully-observable} inspection model, one can efficiently find the optimal randomized IC inspection scheme for instances with  \emph{submodular} inspection cost functions, given value query access to the inspection cost function. 
\medskip

The key ingredients for finding the optimal randomized inspection schemes for submodular inspection cost functions are outlined in the following steps:
 We first devise an algorithm that given a sequence of marginal probabilities $p(j)$ for each action $j$,  computes the distribution over subsets of actions that minimizes the expected inspection cost while inspecting each action $j$ with a marginal probability $p(j)$.
 In particular, we show that this distribution is supported by a few sets (at most $n+1$), and all sequences of marginals with the same order can be supported by the same sets. This means that, for two sequences of marginals $\{p_1(j)\}_{j\in A}, \{p_2(j)\}_{j\in A}$, with a bijection $\pi :[n] \rightarrow A $, for which $p_1(\pi(1)) \leq \ldots \leq p_1(\pi(n))$ and $p_2(\pi(1)) \leq \ldots \leq p_2(\pi(n))$, then the two distributions returned by the  aforementioned algorithm on $p_1,p_2$ have the same support.
 Utilizing this tool,
 for each suggested action $i$, we can partition the interval $[0,1]$ of potential values of $\alpha$ (the payment) into at most $n^2$ sub-intervals for which within each sub-interval, the optimal randomized inspection scheme only inspects the same $n+1$ inspected sets.
 For each such interval, we show that finding the best inspection scheme that suggests action $i$ and uses $\alpha $ in the interval, is equivalent to a simple program in a constant number of variables and constraints.
 Consequently, we can reduce the original problem to solving a small (polynomial) number of optimization problems.

\medskip

We complement the former result and show an impossibility result for the class of XOS inspection cost functions (the next class in the complement-free hierarchy \cite{LehmannLN06} see Section~\ref{sec:model} for a formal definition). 

\medskip

\noindent \textbf{Main Result 2 (Theorem~\ref{thm:impossibility}):} In the \emph{fully-observable} inspection model, finding the best randomized IC inspection scheme for instances with \emph{XOS} inspection cost functions, requires an exponential number of value (or demand) queries. 

\medskip

To demonstrate this hardness result, we carefully construct a family of instances with the same values of $A,f,c$, but with different XOS inspection cost functions parameterized by a set $T$. 
For each inspection cost function in the family, we show that the optimal randomized inspection scheme is a solution to a specific form of optimization problem that always admits a unique optimal solution.
Distinguishing from the inspection cost function what $T$ is, requires an exponential number of value and demand queries. Given the knowledge of $T$, one can easily find the unique optimal inspection scheme, and from the unique optimal inspection scheme, one can easily derive the parameter $T$ of the instance.  We, furthermore, show an stronger inapproximability result for the broader class of subadditive inspection costs:

\medskip

\noindent \textbf{Main Result 3 (Theorem~\ref{thm:inapproximability}):} In the \emph{fully-observable} inspection model, finding a $\nicefrac{22}{21}$-approximation to the best randomized IC inspection scheme for instances with \emph{subadditive} inspection cost functions, requires an exponential number of value (or demand) queries.

\medskip

In Section~\ref{sec:partial-model}, we extend our results in the \textit{partially-observable} inspection model. We first show that this model generalizes the (single-agent) model that was considered concurrently and independently in the work\footnote{We note that \citet{FallahJ23} also considered settings with multiple agents which is not captured by our model.} of \citet{FallahJ23}. We also show that, in contrast to the \textit{fully-observable} inspection model, deterministic inspection schemes can only be found efficiently for submodular functions (see Theorem~\ref{thm:det-partial-positive}) and not all set functions (see Proposition~\ref{prop:det-partial-negative}). In addition, we show that the problem of finding a randomized IC inspection schemes in the \textit{partially-observable} inspection model admits a convexity property, which allows approximating it:

\medskip

\noindent \textbf{Main Result 4 (Theorem~\ref{thm:insp-partial}):} In the \emph{partially-observable} inspection model,
one can efficiently find an additive $\varepsilon$-approximation to the optimal randomized IC inspection scheme for instances with \emph{submodular} inspection cost functions, given value query access to the inspection cost function. 

\medskip

Moreover, in Appendix~\ref{app:gap}, we extend Example~\ref{ex:gap} and show that the gap between 
the optimal deterministic inspection scheme and optimal randomized inspection scheme, can be as large as $\Omega(n)$. 
Since the utility of the principal is bounded by the welfare achievable in the non-strategic version of the problem, the gap is known to also be at most $O(n)$. Lastly, in Appendix~\ref{app:non-ic} we discuss the gaps in the principal's utility between IC inspection schemes, and non-IC inspection schemes for both the deterministic and randomized cases. In particular, we show that there is no gap in the case of deterministic inspection schemes, but there are cases that admit a gap for randomized inspection schemes.

\subsection{Related Work}

\paragraph{(Algorithmic) Contract Theory.} After the seminal work of \citet{GrossmanH83}, who initiated the field of ``Contract Theory'', and that of \citet{HolmstromM87, HolmstromM91}, recent years have witnessed a surge of interest in algorithmic questions related to this model \citep{DuttingRT19,DuttingEFK21,DuttingRT21,DuttingEFK23}. \citet{DuttingRT19, DuttingRT21}  are among the first to model the problem from an algorithmic viewpoint. Their results prove that it is possible to achieve optimality in the worst case through linear contracts, contracts that pay the agent proportionally to the reward associated with an outcome. Numerous other studies have indicated robust optimality of linear contracts in non-Bayesian settings, particularly in the context of max-min criteria \cite{Carrol15, YuK20, WaltonC22, DaiT22}. Recently, \citet{Kambhampati23} has shown how, in certain scenarios, randomized contracts can enhance principal's reward under max-min criteria. For more recent developments on algorithmic contract design, see \cite{DBLP:journals/fttcs/DuttingFT24}.

\paragraph{Combinatorial Models of Contract Theory.} Mostly related to our setting are works that study algorithmic contract design in various combinatorial settings. \citet{DuttingEFK21} consider agents who can select a subset of actions instead of being limited to a single action, and the probability of success is a function of the set of taken actions. They show that if the success probability function is Gross-Substitutes (a set function class defined in \cite{KelsoC82}), then the optimal contract can be found in polynomial time. This model has been adapted by \cite{Dutting24, Deocampo24,ezra_et_al:LIPIcs.ITCS.2024.44}, for different classes of combinatorial set functions including submodular, XOS, and supermodular.  

Another line of research extended the principal-agent model to accommodate settings with multiple agents \cite{BabaioffFNW12,EmekF12,DuttingEFK23,CastiglioniM023, Deocampo24,ezra_et_al:LIPIcs.ITCS.2024.44}. Pioneering work on combinatorial agency by \citet{BabaioffFNW12} and \citet{EmekF12} considers a model where $n$ agents exert effort (or not) and succeed in their assigned tasks with some probability. The success or failure of the entire project is then assumed to be a Boolean function of individual successes and failures. 
\citet{DuttingEFK23} draw the connection between multi-agent contract design and the complement-free hierarchy \cite{LehmannLN06}: Here, each agent possesses a binary action space, and the success probability is a function of the set of agents exerting effort. They show that
the frontier of tractability is the XOS function class, as the function becomes subadditive, the problem is hard to approximate.

Conversely, \citet{AlonLST21} investigate the ``common agency'' problem, which revolves around coordinating several principals to propose a contract to an agent. They expand upon VCG contracts, initially introduced by \citet{LaviS22} for situations with complete information, to accommodate scenarios where principals lack complete information about the agent's attributes.

Finally,  \citet{CastiglioniM023} considers scenarios where the principal observes not just the final outcome, but also the individual (stochastic) outcomes of the agents. Then, the goal of the principal is to design contracts that depend on the individual outcomes. Their focus is on principal's utility functions with the properties of IR-supermodularity (increasing returns), i.e., her utility grows faster as agents put more effort, or DR-submodularity (diminishing returns), i.e., her utility grows slower.

\paragraph{Contracts in Bayesian and Learning settings.} \citet{AlonDT21} explore the contract design problem under private cost per unit-of-effort model, and \citet{AlonDLT23, CastiglioniMG22, GottliebM22} in Bayesian settings (with private types). Moreover, in the Bayesian setting, \citet{GuruganeshSW21} and \citet{CastiglioniMG22-2} demonstrate that the principal can enhance her utility by offering a randomized menu of contracts rather than a single fixed contract. 

In the context of online learning, \citet{HoSV16} examine the learning of monotone contracts, presenting a zooming algorithm that achieves sublinear regret. Relaxing the monotonicity assumption, \citet{ZhuBYWJJ23} consider the intrinsic dimension as a complexity measure based on the covering number of a spherical code. They derive a regret bound that exponentially depends on the intrinsic dimension.

\paragraph{Inspection Strategies.} In the context of inspections, the recent independent and concurrent work of \citet{FallahJ23} study the problem of a single principal interacting with multiple agents: In particular, the principal aims to encourage safety-compliant actions to maximize utility by utilizing payments and a constrained budget for random inspections.  
Finally, \citet{BallK23} have investigated the problem of determining the optimal timing for inspections in a time-dependent setting. Focusing on scenarios where agents exhibit impatience over time, the study explores the ideal moments for conducting inspections.

\section{Fully-Observable Inspection Model}\label{sec:model}
We consider a setting of a single principal and a single agent with binary outcomes. The agent can select an action $a \in A$, each is associated with a cost $c(a)\in \R_{\geq 0}$ and a success probability $f(a) \in [0,1]$. If the agent selects action $i$, the task is successful with probability $f(a)$, which gives the principal a reward (which we normalize without loss of generality to $1$). 
We assume that the agent has a null action $\bot \in A$ for which  $c(\bot)=0$ which enforces the game to be individually rational for the agent, and that there are overall $n$ actions (unless stated otherwise). The null action can be viewed as allowing the agent to opt out of the suggested contract.
The principal can inspect whether some actions were taken by the agent and pay for the inspection cost. If the principal inspects action set $S \subseteq A$, then she gets to observe whether \textit{any} action $j \in S$ has been taken by the agent. In Section~\ref{sec:partial-model}, we analyze different inspection model where the principal only observes whether \textit{some} action $j \in S$ has been taken by the agent. The inspection costs the principal $\insp(S)$, where $\insp:2^A \rightarrow \R_{\geq 0}$ is a combinatorial function which is normalized (i.e., $\insp(\emptyset) =0$) and monotone (i.e., for every $S\subseteq T\subseteq  A$ it holds that $\insp(S)\leq \insp(T)$). An instance of our setting is then defined by a quadruplet $(A,c,f,v)$.
The principal's task is to design an inspection scheme consisting of three components: (1) a suggested action $a\in A$, (2) a contract $\alpha \in [0,1]$, and (3) an inspection distribution $p$ which is a distribution over sets that the principal commits to inspect (i.e., $\sum_{S \subseteq A} p(S) =1$, and for all $S\subseteq A$, $p(S)\in [0,1]$). 
An inspection scheme $(a,\alpha,p)$ is deterministic if there exists $S\subseteq A$, such that $p(S) =1$ and for every $T\neq S$, $p(T)=0$. For an action $a\in A$ and an inspection distribution $p$, we denote by $p(a)\eqdef \sum_{S: a\in S} p(S)$. 

\paragraph{Agent's and Principal's Utility.} In order to introduce the agent's expected utility under inspection scheme $(i,\alpha,p)$, we observe that if the agent takes the incentivized action (action $i$), then the agent's utility is $\alpha f(i) - c(i)$ in expectation. If, instead, the agent decides to take action $j \neq i$, as long as the principal does not inspect a set containing either $i$ or $j$, then we say that the agent is not caught and is paid the promised $\alpha$-fraction of the reward (in case of success). Otherwise (the principal inspects some set containing $i$ or $j$), the agent is caught and thus paid $0$. In expectation, we have that, for $j \neq i$, the utility is $\alpha f(j) \cdot (1 - \sum_{S: \{i, j\} \cap S \neq  \emptyset} p(S))- c(j)$.\footnote{Notice that $1 - \sum_{S: \{i, j\} \cap S \neq  \emptyset} p(S)$ is precisely the probability the agent does not get caught by the principal.} To summarize, given an inspection scheme $(i, \alpha, p)$ the principal proposes, the agent's expected utility from selecting action $j$ is:
\begin{align*}
    u_{\agent}(i, \alpha, p,j) = \begin{cases}
       \alpha f(j) - c(j) &\text{ if } j = i\\
       \alpha f(j) \cdot (1 - \sum_{S: \{i, j\} \cap S \neq \emptyset} p(S))- c(j) &\text{ if } j \neq i
    \end{cases}.
\end{align*}
On the other hand, the principal's expected utility from proposing inspection scheme $(i, \alpha, p)$ and the agent selecting action $j$ is
\begin{align*}
    u_{\principal}(i, \alpha, p,j) = \begin{cases}
        \left(1 - \alpha \right) \cdot f(j) - \sum_{S \subseteq A} p(S) \cdot \insp(S) &\text{ if } j = i\\
        \left(1 - \alpha \cdot \left(1 - \sum_{S: \{i, j\} \cap S \neq \emptyset} p(S)\right)\right) \cdot f(j) - \sum_{S \subseteq A} p(S) \cdot \insp(S) &\text{ if } j \neq i
    \end{cases}.
\end{align*}
The goal of the principal is to find an \textit{incentive-compatible} (IC) inspection scheme $(i, \alpha, p)$, where, by incentive-compatible, we mean that 
$i \in \arg\max_{\ell \in A} u_{\agent}(i, \alpha, p,\ell)$, i.e., the agent is incentivized to take the suggested action $i$ (we assume that in case of a tie, the agent breaks the tie in favor of the suggested action). 
Among the incentive-compatible inspection schemes, the principal wants to find the one with the highest utility,
i.e., $u_{\principal}(i, \alpha, p,i)$ is maximized, and $i \in \arg\max_{\ell \in A} u_{\agent}(i, \alpha, p,\ell)$. 

We first observe that it is always better for the principal to inspect the suggested action $i$ by itself or inspect some other set of actions not including $i$.
\begin{observation}\label{obs:prime}
    Given an inspection scheme $(i, \alpha, p)$, define inspection distribution  $p^\prime$ such that $p^\prime(\{i\}) \eqdef p(i)$, $p^\prime(S) \eqdef 0$ for $S$ such that $i \in S, S \neq \{i\}$, and $p^\prime(S) \eqdef p(S)$ otherwise. It holds that, for all $j \in A$,
    \begin{itemize}
        \item $j \in \arg\max_{\ell \in A} u_{\agent}( i, \alpha, p,\ell) \Longleftrightarrow j \in \arg\max_{\ell \in A} u_{\agent}( i, \alpha, p^\prime,\ell)$;
        \item $u_{\principal}(i, \alpha, p^\prime,j) \geq u_{\principal}(i, \alpha, p,j)$.
    \end{itemize}
\end{observation}
\begin{proof}
    The first part of the observation follows since for each $j\in A \setminus \{i\}$, \begin{equation}\label{eq:pp}
        \sum_{S:\{i,j\} \cap S\neq \emptyset} p(S) = \sum_{S:\{i,j\} \cap S\neq \emptyset} p^\prime(S),
    \end{equation} therefore, the functions $u_{\agent}(i,\alpha,p,j) =u_{\agent}(i,\alpha,p^\prime,j)$ (and it also holds that $u_{\agent}(i,\alpha,p,i) = u_{\agent}(i,\alpha,p^\prime,i)$).

    The second part of the claim follows by Equation~\eqref{eq:pp} and by the  monotonicity of $v$, i.e., 
    \begin{eqnarray*}
    \sum_{S \subseteq A} p(S) \cdot v(S) & = & \sum_{S \subseteq A: i \in S} p(S) \cdot v(S)  +\sum_{S \subseteq A: i \notin S} p(S) \cdot v(S) \\ &  \geq & \sum_{S \subseteq A: i \in S}p(S) \cdot v(\{i\}) +\sum_{S \subseteq A: i \notin S} p(S) \cdot v(S) \\ &  = & p(i) \cdot v(\{i\}) +\sum_{S \subseteq A: i \notin S} p(S) \cdot v(S) = \sum_{S \subseteq A} p^\prime(S) \cdot v(S),\end{eqnarray*} 
    which concludes the proof.
\end{proof}

Note that the modification to the inspection scheme of Observation~\ref{obs:prime} above preserves determinism, i.e., if the initial inspection scheme $(i,\alpha,p)$ is deterministic, then also $(i,\alpha,p^\prime)$ created by Observation~\ref{obs:prime} is.

In Appendix~\ref{app:non-ic}, we discuss non-incentive-compatible inspection schemes, where the principal might suggest some action $j$ to induce the agent to best respond with action $i$ in order to gain higher utility. We show that, for deterministic inspection schemes, incentive-compatibility is without loss of generality, and the principal has no gain in suggesting an action that is not the best response of the agent. However, for randomized inspection schemes, non-incentive-compatible inspection schemes can yield higher principal's utility than the best incentive-compatible one.

\paragraph{Combinatorial Inspection Cost Functions.}

We focus on inspection cost functions 
$\insp: 2^A \rightarrow \R_{\geq 0}$
that belong to one of the following classes of set functions \citep{LehmannLN06}:

\begin{itemize}
    \item Set function $\insp$ is \emph{additive} if for every $S\subseteq A$ it holds that $\insp(S)=\sum_{i\in S} \insp(\{i\})$. 
    \item Set function $\insp$ is \emph{submodular} if for any two sets $S, S' \subseteq A$ with $S \subseteq S'$ and any $i \in A$ it holds that $\insp(i \mid S) \geq \insp(i \mid S')$ where $\insp(i \mid S)  = \insp(\{i\} \cup S) - \insp(S)$ is called marginal.
    \item Set function $\insp$ is \emph{XOS} if there exists a collection of additive functions $\{\gamma_\ell: 2^A \rightarrow \R_{\geq 0}\}_{\ell \in [k]}$ such that for each set $S \subseteq A$ it holds that $\insp(S) = \max_{\ell \in [k]} \gamma_\ell(S) $.  
    \item Set function $\insp$ is \emph{subadditive} if for any two sets $S, S' \subseteq A$ it holds that $\insp(S) + \insp(S') \geq \insp(S \cup S')$.
\end{itemize}

It is well known that $\text{additive} \subset \text{submodular} \subset \text{XOS} \subset \text{subadditive}$ and all containment relations are strict \citep{LehmannLN06}.

\paragraph{Oracles for accessing $v$.}
As is common in the combinatorial optimization literature involving set functions, we assume two primitives for accessing $v$:
\begin{itemize}
    \item A \emph{value oracle} for $v$ is given $S \in 2^A$ and returns $v(S)$.
    \item A \emph{demand oracle} for $v$ is given a vector of prices $q = (q_1, \ldots, q_n) \in \R^n_{\geq 0}$ and returns a set $S \in 2^A$ that maximizes $v(S) - \sum_{j \in S} q_j$. We denote this set by $D(v, q)$.
\end{itemize}

Both value and demand oracles are considered standard in combinatorial optimization problems over set functions. 
In markets for goods (e.g., combinatorial auctions), a demand query corresponds to the best bundle to purchase given item prices. 
Demand oracles have proven useful in previous studies on combinatorial contracts \citep{DuttingEFK21} and multi-agent contracts \citep{DuttingEFK23}. 
We assume that all values have a description of size $O(1)$, as our goal is to understand the (computational and information) complexity of the problem as the number of actions grows large.

\section{Deterministic Inspection Schemes}\label{sec:deterministic}

In this section, we present a polynomial time algorithm that finds the best contract with a deterministic inspection scheme. Formally, we prove the following theorem (see proof in Appendix~\ref{app:omitted}):
\begin{restatable}{theorem}{thmdet}\label{thm:det}
There exists an algorithm that for every instance $(A,c,f,v)$ returns the optimal \textit{deterministic} IC inspection-scheme that runs in time polynomial in $n=|A|$, and uses at most $ n^2$ value queries to $v$.
\end{restatable}

In order to prove the above theorem, we first note that a deterministic inspection scheme can be described by a triplet $(i, \alpha, S)$ composed of the suggested action $i$, the contract $\alpha$, and the inspected set $S$. We next present the algorithm for finding the best deterministic inspection scheme. 

\paragraph{Algorithm.}
We first go over all actions with $c(i)=0$, and consider the one with the maximal value of $f(i)$, i.e., let $i^* = \arg\max_i \{f(i) \mid i\in A \wedge c(i)=0\}$.
This action can be incentivized by using the inspection scheme $(i^*,0,\emptyset)$, which gives the principal an expected utility of $f(i^*)$.

We then go over all actions $i\in A$ with $f(i)>c(i)>0$, and for each such $i$ we define\footnote{Observe that $A_i$ is well defined since $f(i)>f(j)$ and $f(i)>0$.} $$A_i = \left\{ j \mid f(j) < f(i) \wedge \frac{c(i)-c(j)}{f(i)-f(j)}  > \frac{c(i)}{f(i)} \right\},$$ 
in particular, for every $j \in A_i$ it must be that $\frac{c(i)-c(j)}{f(i)-f(j)} > \frac{c(i)}{f(i)} >0 $. The interpretation of $A_i$ is the set of actions $j$ with lower value of $f$ than $i$, such that if the principal uses $\alpha = \frac{c(i)}{f(i) }$ then the agent strictly favors them over $i$. 
Then, we define $$S_{i}\eqdef A_i \cup \left\{ j^\prime \neq i \mid f(j^\prime) \geq f(i) \wedge f(j^\prime) \cdot \frac{c(i)}{f(i)} >c(j^\prime) \right\} $$ and for each $j\in A_i$, we define 
\begin{align*}
	S_{i,j} \eqdef & \left\{ j^\prime\in A_i \mid \frac{c(i)-c(j^\prime)}{f(i)-f(j^\prime)} > \frac{c(i)-c(j)}{f(i)-f(j)} \right\} \\ &\cup \left\{ j^\prime  \mid f(j^\prime) \geq f(i) \wedge \frac{c(i)-c(j)}{f(i)-f(j)} \cdot (f(j^\prime)-f(i))  > c(j^\prime)-c(i) \right\}.
	\end{align*}
The interpretation of $S_i$ (respectively, $S_{i,j}$) is the set of actions that needs to be inspected if $i$ is not inspected and $\alpha= \frac{c(i)}{f(i)}$ (respectively, when $\alpha =\frac{c(i)-c(j)}{f(i)-f(j)} $), as if any of these actions is not inspected, the agent strictly favor them over action $i$.

We then consider  for each $j\in A_i $ the inspection scheme $(i,\frac{c(i)-c(j)}{f(i)-f(j)}, S_{i,j})$.
We also consider the inspection schemes $(i,\frac{c(i)}{f(i)}, S_i) $, and $(i,\frac{c(i)}{f(i)}, \{i\}) $.
We return the inspection scheme with the maximum utility for the principal among all considered inspection schemes.

\section{Randomized Inspection Schemes: Submodular Inspection Cost} \label{sec:submod}

In this section, we shift our attention to randomized inspection schemes, recalling that the gap between deterministic and randomized inspection schemes in terms of the principal's expected utility can be as large as $\Omega(n)$ (see Appendix~\ref{app:gap} for further details). 
Our main result is a polynomial time algorithm that finds the optimal randomized inspection scheme for submodular inspection cost functions. Formally, we show  the following:

\begin{theorem}\label{thm:submod}
    For every instance $(A,c,f,v)$ with submodular inspection cost function $v$, there exists an optimal \textit{randomized} incentive-compatible inspection scheme $(i,\alpha,p)$ where $p$ is supported by at most $n+1$ sets. Moreover, such an optimal inspection scheme can be found in time polynomial in $n=|A|$ using value query access to $v$.
\end{theorem}

Before describing the algorithm and showing why Theorem~\ref{thm:submod} holds, we present several definitions and properties of submodular functions.

\begin{definition}
    Given a function $v:2^{A} \rightarrow \R$, the Lov\'asz extension of $v$ is  $v^L:[0,1]^A \rightarrow\R$ where $$ v^L(x) = \mathbb{E}_{\lambda \sim U[0,1]}[v(\{i\in A \mid x_i \geq \lambda \})].$$
\end{definition}
It is known that the Lov\'asz extension of $v$ is the convex closure of $v$ which implies the following lemma (for which we add a formal proof in Appendix~\ref{app:omitted} for completeness).
The next lemma states that, given the marginals of the inspection distribution $p(j)$, if the inspection cost function is submodular, then one can find the distribution with the minimum expected inspection cost consistent with those marginals.

\begin{restatable}{lemma}{lemsubmod}\label{lem:submod}
    For a set $A$ with cardinality $n$, let $p:2^{A} \rightarrow \R_{\geq 0}$ be an arbitrary function, and let $v:2^{A}\rightarrow \R_{\geq 0}$ be a submodular function. Let $p(j) = \sum_{S \subseteq A:j \in S} p(S)$. Let   $\pi:[n] \rightarrow A$ be an arbitrary bijection such that   $p(\pi(1)) \leq \ldots \leq p(\pi(n))$ (such a bijection always exists). For $p^\prime:2^{A} \rightarrow \R_{\geq 0}$ such that for every $t \in [n]$, $p^\prime(\{\pi(t),...,\pi(n)\}) = p(\pi(t))-p(\pi(t-1))$ where $p(\pi(0))=0$, otherwise $p^\prime(S) =0$ for $S\neq \emptyset$, and $p^\prime(\emptyset) = \sum_{S\subseteq A} p(S)-p(\pi(n))$  it holds that:
    \begin{enumerate}
    \item $\sum_{S \subseteq A}  p^\prime(S) = \sum_{S \subseteq A} p(S)$;
        \item For all $j \in A$,  $p^\prime(j) =\sum_{S \subseteq A:j \in S} p^\prime(S)= p(j)$;
        \item $\sum_{S \subseteq A} p^\prime(S) \cdot v(S) \leq \sum_{S \subseteq A} p(S) \cdot v(S)$.
    \end{enumerate}
\end{restatable}

\subsection{Algorithm}\label{sec:alg}

Our algorithm reformulates the problem of finding the best inspection scheme as solving a polynomial number of optimization problems, each with a constant number of variables and constraints. We first write it as a program over an exponential number of variables (the probabilities to inspect each set $p(S)$, and the payment $\alpha$). We then partition the interval $[0,1]$ of options for $\alpha$ to sub-intervals such that within each sub-interval the marginal probabilities that the actions are needed to be inspected have the same relative order.
This allows us to solve the program in each sub-interval separately, and use Lemma~\ref{lem:submod} to decrease the number of variables from the objective function from an exponential number to a constant number. We then solve each optimization (non-linear) program, and return the best one. We now explain our algorithm formally:

\paragraph{Step 1: Going over the actions.} We fix action $i \in A$. If $c(i) = 0$, then the inspection scheme $(i, 0, p(\emptyset))$, where $p_\emptyset$ is the deterministic inspection distribution that inspects $\emptyset$ with probability $1$, is the optimal IC inspection scheme that incentivizes action $i$ (since the principal pays $0$). 

If instead $f(i) \leq c(i)$, then incentivizing action $i$ cannot guarantee a positive utility to the principal, and therefore, the inspection scheme $(\bot,0,p_\emptyset)$ guarantees at least the same utility. 

Otherwise, we have that $f(i) > c(i) > 0$, then by using Observation~\ref{obs:prime} we can formulate the problem of finding the optimal (randomized) IC inspection scheme that incentivizes action $i$ as the following quadratic program:
\begin{align}
    \min_{\alpha, p} \ &\alpha f(i) + \sum_{S \subseteq A} p(S) \cdot \insp(S) \nonumber\\
    \text{s.t. } &\alpha f(i) - c(i) \geq \alpha f(j) \cdot \left(1 - p(i) - p(j)\right) - c(j), \forall j \neq i \label{eq:insp}\\
    \qquad &\alpha \in [0,1], p \in \Delta(2^A), p(i)=p(\{i\}) \nonumber.
\end{align}
We note that the constraints written above are the incentive-compatibility constraints, written in terms of marginal probabilities of inspection $p(j) = \sum_{S: j\in S} p(S)$, that we need to satisfy for the agent to take action $i$. 
Since $\bot \in A$, we can add a constraint to Program~\eqref{eq:insp} of the form $\alpha \geq \frac{c(i)}{f(i)}$ without changing feasibility. Next,  for each $j \in A$, if  $f(j) = 0$, the incentive-compatibility constraint is less restrictive than $\alpha \geq \frac{c(i)}{f(i)}$ and can thus be discarded. Else ($f(j)>0$),  the incentive-compatibility constraint is equivalent to
\[
    p(j) \geq 1 - p(i) - \frac{\alpha f(i)- c(i) + c(j)}{\alpha f(j)}.
\]
For convenience, we name the right-hand side above as $\eta_j(\alpha, p(i))$, for each $j \neq i$ with $f(j)>0$.

\paragraph{Step 2: Partitioning the interval $[0,1]$   to sub-intervals.} Let us now define the set of $\alpha$'s for a specific action $i$ and each of the $n-1 \choose 2$ pairs $j, j^\prime \neq i$ as the $\alpha$'s respectively satisfying $\eta_j(\alpha, p(i)) = \eta_{j^\prime}(\alpha, p(i))$. Formally, let $C_i$ be the set of all $\alpha_i(j,j^\prime)$'s such that $f(j) \neq f(j^\prime)$ and that satisfy
\[
    \alpha_i(j,j^\prime):~ \eta_j(\alpha, p(i)) = \eta_{j^\prime}(\alpha, p(i)) \Longleftrightarrow \alpha_i(j,j^\prime) \eqdef \frac{(c(i) - c(j)) \cdot f(j^\prime) - (c(i) - c(j^\prime)) \cdot f(j)}{(f(j^\prime) - f(j)) \cdot f(i)}.
\]
Let us observe that $\alpha$'s so defined do not depend on $p(i)$. Hence, without loss of generality, we order elements in $C_i$ in non-decreasing order and reindex them as $\alpha_{i, \ell}$'s. We first remove all $\alpha_{i,\ell}$'s not in $(0,1)$, and we add $\alpha_{i,0}=0$, and $\alpha_{i,|C_i|}=1$ to $C_i$. 

\paragraph{Step 3: Solve the simplified program.} We note that within any interval $[\alpha_{i,\ell}, \alpha_{i, \ell+1})$, the order of $\eta$'s is preserved, that is, there exists a bijection $\pi_{i, \ell}: [n-1] \rightarrow A \setminus \{i\}$ such that $\eta_{\pi_{i, \ell}(1)}(\alpha, p(i)) \leq \ldots \leq \eta_{\pi_{i, \ell}(n-1)}(\alpha, p(i))$. For convenience, we set $\eta_{\pi_{i, \ell}(0)}(\alpha, p(i)) = -\infty$ and $\eta_{\pi_{i, \ell}(n)}(\alpha, p(i)) = +\infty$.

Moreover, since this is a minimization problem, we would like $p(j)$'s to be as small as possible (yet non-negative), thus, we can derive that $p(j) = \max\{0, \eta_j(\alpha, p(i))\}$ for all $j \in A \setminus \{i\}$. 

Let $k \in \{0\} \cup [n-1]$ be such that $\eta_{\pi_{i, \ell}(k)}(\alpha, p(i)) \leq 0 < \eta_{\pi_{i, \ell}(k+1)}(\alpha, p(i))$ and,  by applying Lemma~\ref{lem:submod} on the set $A \setminus \{i\}$, 
we write Program~\eqref{eq:insp} in the following simplified form:
\begin{align}
    \min_{\alpha, p(i)} \ &\alpha f(i) + p(i)\cdot v(\{i\}) + \eta_{\pi_{i, \ell}(k+1)}(\alpha, p(i)) \cdot v(A \setminus \{i, \pi_{i, \ell}(1), \ldots, \pi_{i,\ell}(k)\}) \nonumber\\
    & + \sum_{t=k+2}^{n-1} (\eta_{\pi_{i, \ell}(t)}(\alpha, p(i)) - \eta_{\pi_{i, \ell}(t-1)}(\alpha, p(i))) \cdot v(A \setminus \{i, \pi_{i, \ell}(1), \ldots, \pi_{i,\ell}(t-1)\}) \nonumber\\
    \text{s.t. } &\eta_{\pi_{i, \ell}(k)}(\alpha, p(i)) \leq 0 < \eta_{\pi_{i, \ell}(k+1)}(\alpha, p(i)) \label{eq:insp-simple}\\
    \qquad &\max\left\{\alpha_{i, \ell}, \frac{c(i)}{f(i)}\right\} \leq \alpha \leq \alpha_{i, \ell+1} \nonumber.
\end{align}
Let us observe that Program~\eqref{eq:insp-simple} is a constrained minimization problem in two variables, namely, $\alpha, p(i)$ which we are able to solve due to its specific structure. In particular, we utilize the fact that there are only two variables, and one of them only appears linearly in the objective function. 
More specifically, one can define $\beta = 1/\alpha$. Then, the program is minimizing a function of the form $C_1/\beta + C_2 \cdot p(i) + C_3\cdot \beta + C_4$ (where $C_1,C_2,C_3,C_4$ are constants) subject to $\beta, p(i)$ belonging to a polygon with four sides (the four constraints). Since $p(i)$ only appears as a linear term in the objective function, we know that the optimum will be on the edges of the polygon. Thus, for each such edge, we rewrite the objective function as minimizing a function of one variable in the form of a quadratic numerator over a linear denominator. For this minimization, there is a solution in closed form.

\paragraph{Step 4: Output.} The algorithm considers all the inspection schemes in Steps 1, 2, 3 above for each $i \in A$, $\ell \in {n-1 \choose 2}$, $k \in \{0\} \cup [n-1]$. For each, it computes, using Program~\eqref{eq:insp-simple}, the optimal $\alpha, p(i)$, and calculates the optimal inspection scheme using Lemma~\ref{lem:submod} with $p(j) =0 $ for $j \in \{\pi_{i,\ell}(1),\ldots,\pi_{i,\ell}(k)\},$  $p(j) = \eta_j(\alpha, p(i)) $ for $j \in \{\pi_{i,\ell}(k+1),\ldots,\pi_{i,\ell}(n-1)\},$ and $p(i)$.
It returns the inspection scheme with the maximal expected utility among all suggestions.

\subsection{Proof of Theorem~\ref{thm:submod}}

\paragraph{Correctness.}
It is sufficient to consider actions $i\in A$ with $f(i)>c(i)>0$, in particular, that a feasible solution for Program~\eqref{eq:insp} corresponds to a feasible solution for Program~\eqref{eq:insp-simple}, and vice versa. Consider the optimal inspection scheme that incentivizes action $i$,  which by Observation~\ref{obs:prime} is the solution $\alpha,p$ to Program~\eqref{eq:insp} (with respect to action $i$). 
We know that there exists an $\ell$ such that $\alpha \in [\alpha_{i, \ell},\alpha_{i,\ell+1})$. 
Since within $[\alpha_{i,\ell}, \alpha_{i, \ell+1})$ the order of $\eta$'s is preserved, by the structure guarantee of Lemma~\ref{lem:submod} we obtain that there exists a distribution $p^\prime$ with minimum expected cost with marginals $\max\{0,\eta_{j}\}$, for $j\neq i$, and $p(i)$ for $j=i$ (that always inspect $i$ as a singleton), and is supported by the set $\{i\}$ and on sets of the form $\{\pi_{i, \ell}(j),...,\pi_{i, \ell}(n-1)\}$ for $j\in  [n]$. 
Since $(\eta_{\pi_{i,\ell}(0)}, \ldots,\eta_{\pi_{i,\ell}(n)})$ is an increasing sequence with $\eta_{\pi_{i,\ell}(0)}=-\infty$, and $\eta_{\pi_{i,\ell}(n)}=\infty$, then there exists $k$ such that $\eta_{\pi_{i, \ell}(k)}(\alpha, p^\prime(i)) \leq 0 < \eta_{\pi_{i, \ell}(k+1)}(\alpha, p^\prime(i))$, which means that for the corresponding Program~\eqref{eq:insp-simple} with $i,\ell,k$, the constraints are satisfied (and the objective has the same value).
Overall, the optimal solution $(\alpha,p)$, corresponds to a feasible solution with the same objective value for one of the Programs~\eqref{eq:insp-simple} with the corresponding $i,\ell,k$.

We also need to show that from a feasible solution to Program~\eqref{eq:insp-simple}, we can construct a feasible solution to Program~~\eqref{eq:insp} with the same objective value.
Consider a solutions $\alpha, p(i)$ to Programs~\eqref{eq:insp-simple}, parametrized by combinations of $i, \ell, k$. For marginals of $0$ for all $j \leq k$, and of $\eta_{\pi_{i, \ell}(j)}(\alpha, p(i))$ for $j>k$, we apply Lemma~\ref{lem:submod} (on the set $A\setminus \{i\}$), and create the minimal cost distribution $p^\prime$ where $p^\prime(\{i\})=p(i)$, and for all $j \in  [n-1]$ we have $p^\prime(\pi_{i,\ell}(j))=\max\{0,\eta_{\pi_{i, \ell}(j)}(\alpha, p(i))\}$.\footnote{Observe that Lemma~\ref{lem:submod} can be applied using only the marginals $p(j)$'s instead of receiving a full distribution $p$. Moreover, to ensure that the created distribution $p^\prime$ is indeed a distribution, we can set  $p^\prime(\emptyset)$ to $1-p(i)-\max_{j\neq i} p(j)$ which is at least 0 as otherwise this would not be a feasible solution to Program~\eqref{eq:insp-simple}.}
Thus, $\alpha,p^\prime$ is a feasible solution to Program~\eqref{eq:insp} since for every $j\neq i$, we have that $p^\prime(j) \geq \eta_j(\alpha,p(i))$ (which  equivalent to  the constraint of Program~\eqref{eq:insp}). 
By Lemma~\ref{lem:submod}  $p^\prime$ has a lower expected inspection cost than $p$ and since it uses the same value of $\alpha$,
it concludes the proof.

\paragraph{Runtime and Number of Value Queries.} The algorithm presented in this section contains three nested loops: Step 1 is performed for each $i \in A$, that is $n$ times. For each Step 1, Step 2 is performed at most $O(n^2)$ times, since this is the maximum cardinality $C_i$ has. For each iteration of Step 2, Step 3 solves a program in two variables ($\alpha, p(i)$) for each possible value of $k$, i.e., $n-1$ times at most. Thus, the algorithm runs in polynomial time.

Moreover, as for the number of value queries used by the algorithm, we know from what is argued above that there are $O(n^3)$ iterations of Steps 1, 2. For each different $k$ we iterate over in Step 3, we query the value of the same sets, thereby making $n$ value queries for all the iterations of Step 3. Thus, the algorithm requires $O(n^4)$ value queries to the inspection cost function $v$ overall.

\paragraph{Support.} The support of the created distribution is of size at most $n+1$. This follows from the fact that besides adding the set $\{i\}$ to the support, the distribution is supported by the outcome of Lemma~\ref{lem:submod} on $A \setminus \{i\}$ which is at most of size $n$.

\bigskip

\noindent The theorem follows.

\section{Impossibility for XOS Inspection Cost} \label{sec:xos}
In this section, we complement our positive result from Section~\ref{sec:submod} and show that for the class of XOS inspection cost functions (the next class in the hierarchy of complement-free set functions \cite{LehmannLN06}), it is not possible to find the optimal (randomized)  inspection scheme using a polynomial number of value queries. 
Moreover, the impossibility result applies even in the case where the algorithm has access to function $v$ via demand queries (as opposed to just value queries). Formally, we prove the following:
\begin{theorem}\label{thm:impossibility}
    Every algorithm that for every instance $(A,f,c,\insp) $ with XOS inspection cost function $\insp$ over $n=|A|$ actions returns the optimal randomized inspection scheme, uses at least $\left(\frac{5}{4}\right)^n$ value or demand queries. 
\end{theorem}

To show Theorem~\ref{thm:impossibility}, we construct a family of instances parameterized by a random hidden set $T$ such that, for every instance, there exists a \emph{unique} principal's utility maximizing randomized inspection scheme, and there is a one-to-one correspondence between $T$ and this unique best inspection scheme. Additionally, using less than an exponential number of value and demand queries, one cannot extract the hidden parameter of the instance (set $T$). 

\subsection{Hard Instance Construction} 
Let us consider an instance with $n = k+3$ actions for a $k>5$ being a prime number, where we name the actions $A = \{\bot, g, x, 1, \ldots, k\}$. We now define a family of instances $\inst_T=(A, f, c, v_T)$ parametrized by a random set $T \subseteq A$, where $|T| = \left\lceil\frac{4k}{5}\right\rceil$, and such that: (1) $f(\bot) = 0, ~f(g) = 1, ~f(x) = \frac{3}{10}, ~f(i) = \frac{1}{5} ~\forall i \in [k]$, (2) $ c(\bot)=0$, $c(g) = \frac{1}{10}, ~c(x) = \frac{1}{100}, ~c(i) = \frac{1}{100} ~\forall i \in [k]$. We next describe the inspection cost function\footnote{We note that since $k > 5$ is a prime number, then $4k/5$ is always non-integer.}:
\begin{align}\label{eq:xos}
    v_T(S) \eqdef \frac{\ind{S \setminus \{\bot, g\} \neq \emptyset}}{40} &+ \ind{\bot \in S } +\ind{g \in S } \\
    &+ \frac{1}{80k} \cdot \begin{cases} 0 & \text{ if } |S|=0\\ 
        1 &\text{ if } 0<|S \setminus\{\bot, g, x\}| < \frac{4k}{5} \text{ or } S \setminus\{\bot, g, x\} \in \cyclic(T) \\
        2 &\text{ if } |S \setminus\{\bot, g, x\}| > \frac{4k}{5} \text{ and } S \setminus\{\bot, g, x\} \notin \cyclic(T)
    \end{cases},\nonumber
\end{align}
where $\cyclic(T) \eqdef \left\{ T_t \mid t\in [k] \right\}$, and $T_t \eqdef \{((j +t) \mod k)+1 \mid j \in T\}$,
 is defined to be the collection of all sets obtained by cyclic shifts of elements in $T$.

We now show that these inspection cost functions are XOS:

\begin{lemma}\label{lem:xos}
    For every $T \subseteq A$ with $|T| = \left\lceil\frac{4k}{5}\right\rceil$, $v_T$ is XOS.
\end{lemma}
\begin{proof}
    Let us consider the following collection $\Gamma$ of additive functions:
    \begin{itemize}
        \item $\gamma_x(S) \eqdef \frac{1}{40} \cdot \ind{x \in S} + \ind{\bot \in S} + \ind{g \in S}$, defined for action $x$;
        \item $\gamma_i(S) \eqdef \left(\frac{1}{40} + \frac{1}{80k}\right) \cdot \ind{i \in S} + \ind{\bot \in S} + \ind{g \in S}$, defined for each action $i \in [k]$;
        \item $\gamma_{S^\prime}(S) \eqdef \left(\frac{1}{40} + \frac{1}{40k}\right) \cdot \frac{|S \cap S^\prime|}{|S^\prime|} + \ind{\bot \in S} + \ind{g \in S}$, defined for each subset of actions $S^\prime \subseteq [k]$ such that $|S^\prime| > \frac{4k}{5}$ and $S^\prime \notin \cyclic(T)$.
    \end{itemize}
    We now show that (1) for all $S \subseteq A$ and all functions $\gamma \in \Gamma$, $v_T(S) \geq \gamma(S)$, and (2) for all $S \subseteq A$, there exists a function $\gamma \in \Gamma$ such that $v_T(S) = \gamma(S)$. This trivially holds for $S = \emptyset$. Moreover, for both claims, since the summands $\ind{\bot \in S }, \ind{g \in S }$ are common between $v_T$ and $\gamma$'s, we can ignore them, i.e., consider $v^\prime_T$ and $\gamma^\prime$ which are respectively obtained from $v_T$ and $\gamma$ by subtracting the term $\ind{\bot \in S} + \ind{g \in S}$, and consider a new instance without actions $\bot, g$.

    For (1), we observe that, for $S=\{x\}$, 
    it holds that for every $\gamma^\prime 
    \in \Gamma^\prime$, $ \gamma^\prime(S)\leq \frac{1}{40} = v^\prime_T(S)$.
    For all $S \subseteq A\setminus\{\bot, g\}$, such that  $v^\prime_T(S) = \frac{1}{40}+\frac{1}{80k}$, then  $\max\{\gamma^\prime_x(S), \max_{i \in [k]} \gamma^\prime_i(S)\} \leq \frac{1}{40}+\frac{1}{80k} = v^\prime_T(S)$. Also, for every $\gamma^\prime$ of the third type defined by a set $S^\prime$, it holds that $S^\prime\setminus S\neq \emptyset$. Therefore, $\gamma^\prime_{S^\prime}(S) \leq (\frac{1}{40}+\frac{1}{40k})\cdot\frac{|S^\prime|-1}{|S^\prime|} \leq (\frac{1}{40}+\frac{1}{40k})\cdot\frac{\lceil\frac{4k}{5}\rceil-1}{\lceil\frac{4k}{5}\rceil} \leq v_T^\prime(S)$.   
    For $S$ such that $v^\prime_T(S) = \frac{1}{40}+\frac{1}{40k}$, then for all $\gamma^\prime \in \Gamma^\prime$, it holds that $\gamma^\prime(S) \leq \frac{1}{40}+\frac{1}{40k}$, which concludes the proof of (1).
    
    For (2), if $S = \{x\}$ then $\gamma^\prime_x(S) = v^\prime_T(S)$.   Otherwise, if $S$ is such that $v^\prime_T(S) = \frac{1}{40} + \frac{1}{80k}$, then it means that there exists some action  $i \in S \cap [k]$,  therefore, $\gamma^\prime_i(S) = v^\prime_T(S)$. If instead, $v^\prime_T(S) = \frac{1}{40} + \frac{1}{40k}$, then it means that $S$ is such that $|S \setminus\{x\}| > \frac{4k}{5} \text{ and } S \setminus\{x\} \notin \cyclic(T)$, therefore $\gamma^\prime_{S}$ exists (in $\Gamma^\prime$), thus,  $\gamma^\prime_{S}(S) = \frac{1}{40} + \frac{1}{40k}$, which concludes the proof of (2).
\end{proof}

\subsection{Proof of Theorem~\ref{thm:impossibility}}

We next prove Theorem~\ref{thm:impossibility}, as per the following lemmas.

\begin{lemma}\label{lem:unique}
    For every $T \subseteq A$ with $|T| = \left\lceil\frac{4k}{5}\right\rceil$, the unique principal's utility maximizing randomized inspection scheme of instance $\cI_T$ is $\left(g, \alpha=\frac{c(g)}{f(g)}, p\right)$, where $p(S \cup \{x\}) = \frac{1}{2|T|}$, for $S \in \cyclic(T)$, $p(\{x\}) = \frac{2}{3} - \frac{k}{2|T|}$, and $p(\emptyset) = \frac{1}{3}$.
\end{lemma}
\begin{proof}
    In order to show that the unique optimal inspection scheme is $\cI_T$, we first need to show that it is never more profitable for the principal to incentivize some other action $j$ over action $g$.
    To do so,  we bound the utility that the principal can extract from any other action, by the social welfare induced by this action, ignoring incentive-compatibility constraints. In particular, observe that the social welfare of actions $\bot$, $x$ or $i \in [k]$ is respectively $f(\bot) - c(\bot) = 0$, $f(x) - c(x) = \frac{29}{100}$ and $f(i) - c(i) = \frac{19}{100}$. 

    On the other hand, under inspection scheme $\cI_T$, we have that $p(g) = 0, p(\bot) = 0, p(x) = \frac{2}{3}$ and $p(i) = \frac{1}{2}$. Thus, the agent's best response is to take action $g$ since
    \begin{align*}
        u_\agent\left(g, \frac{c(g)}{f(g)}, p, \bot\right) &= f(\bot) \cdot \frac{c(g)}{f(g)} \cdot (1 - p(\bot) - p(g)) - c(\bot) = 0\\
        u_\agent\left(g, \frac{c(g)}{f(g)}, p, x\right) &= f(x) \cdot \frac{c(g)}{f(g)} \cdot (1 - p(x) - p(g)) - c(x) = \frac{1}{100} - \frac{1}{100} = 0\\
        u_\agent\left(g, \frac{c(g)}{f(g)}, p, i\right) &= f(i) \cdot \frac{c(g)}{f(g)} \cdot (1 - p(i) - p(g)) - c(i) = \frac{1}{100} - \frac{1}{100} = 0, ~\forall~i \in [k]\\
        u_\agent\left(g, \frac{c(g)}{f(g)}, p, g\right) &= f(g) \cdot \frac{c(g)}{f(g)} - c(g) = 0,
    \end{align*}
    thus, the agent does not strictly favor any other action over  $g$. Therefore, the principal's utility is
    \begin{align*}
        u_{\principal}(g, \alpha, p,g) &= \left(1 - \frac{c(g)}{f(g)}\right) - \left(p(\{x\}) \cdot v(\{x\}) + \sum_{S \in \cyclic(T)} p(S \cup \{x\}) \cdot v(S \cup \{x\})\right)\\
        &= \frac{9}{10} - \left(\frac{2}{3} - \frac{k}{2|T|}\right)\cdot\frac{1}{40} - \frac{k}{2|T|}\left(\frac{1}{40} + \frac{1}{80k}\right) = \frac{53}{60} - \frac{1}{160 \cdot |T|}.
    \end{align*}
    From above, we see that incentivizing $g$ through inspection scheme $\cI_T$ is always more profitable than incentivizing any other action $j \neq g$, since $\frac{53}{60} - \frac{1}{160 \cdot |T|} > \frac{29}{100} > \frac{19}{100} > 0$ (respectively, the principal's utilities arising by incentivizing actions $x, i\in [k], \bot$).

    Our next step is to show that no other inspection scheme $\cI^\prime_T \eqdef (g, \alpha^\prime, p^\prime) \neq \cI_T$ can be optimal. We can assume that $\alpha^\prime \geq \frac{c(g)}{f(g)}$, as otherwise $g$ would not be incentivized since the agent's utility from choosing action $g$ is strictly negative, but from selecting $\bot$ is $0$. 
    We can also assume that $p^\prime(\bot) = 0$ because it is never profitable to inspect the null action (since it has a strictly positive cost, and the constraint derived from $\bot$ is that the agent's utility from $g$ is non-negative). Moreover, it must be the case that $p^\prime(\{g\})=0$, or else there would exist a distribution $p^{\prime\prime}$ such that $p^{\prime\prime}(\{g\})=0$, $p^{\prime\prime}(\{x\} \cup [k]) = p^{\prime}(\{x\} \cup [k]) + p^\prime(\{g\})$ and $p^{\prime\prime}(S) = p^\prime(S)$ for all other sets, that still incentivizes $g$ but yields better utility to the principal (since $v_T(\{x\}\cup [k]) < v_T(\{g\})$), contradicting optimality of $\cI^\prime_T$. From the incentive-compatibility constraints, we derive that
    \begin{align}\label{eq:constraint-x}
        p^\prime(x) &\geq 1 - \frac{\alpha^\prime f(g)-c(g)+c(x)}{\alpha^\prime f(x)}\\
        p^\prime(i) &\geq 1 - \frac{\alpha^\prime f(g)-c(g)+c(i)}{\alpha^\prime f(i)},~\forall~i \in [k] \nonumber.
    \end{align}
    Therefore, since $f(i), c(i)$ are constant across all actions $i \in [k]$, we obtain
    \begin{align}\label{eq:constraint}
        &k \cdot \left(1 - \frac{\alpha^\prime f(g)-c(g)+c(i)}{\alpha^\prime f(i)}\right) \leq \sum_{i \in [k]} p^\prime(i) = \sum_{i \in [k]}\sum_{S \subseteq A: i \in S} p^\prime(S) \nonumber\\
        &\leq (|T|-1) \cdot \underbrace{\sum_{S \subseteq A: 0 < |S \cap [k]| < \frac{4k}{5}} p^\prime(S)}_{=: ~y_1} +~ |T| \cdot \underbrace{\sum_{S \subseteq A: S \cap [k] \in \cyclic(T)} p^\prime(S)}_{=: ~y_2} +~ k \cdot \underbrace{\sum_{\substack{S \subseteq A: |S \cap [k]| \geq \frac{4k}{5}\\ S \cap [k] \notin \cyclic(T)}} p^\prime(S)}_{=: ~y_3},
    \end{align}
    where the last inequality follows from the definition of marginal probabilities, from partitioning sets into the respective types and upper bounding their cardinalities. We observe that (1) if the principal inspects  $x$ it pays at least $\frac{1}{40}$; (2) if the principal decides to inspect actions from $[k]$, then she pays an additional term that depends on the set of either $\frac{1}{80k}$ or $\frac{1}{40k}$. All in all, the principal's cost (payment plus the inspection cost) is at least
    \begin{align}\label{eq:objective}
        \alpha^\prime f(g) + \frac{1}{40} \cdot p^\prime(x) + \frac{1}{80k} \cdot (y_1 + y_2) + \frac{1}{40k} \cdot y_3.
    \end{align}
    In order to lower bound the principal's cost, we simply need to minimize \eqref{eq:objective} subject to constraints in \eqref{eq:constraint-x} and \eqref{eq:constraint} as well as the constraints of $\alpha' \geq \frac{c(g)}{f(g)}$, $y_1+y_2+y_3\leq 1$ and $0\leq p'(x),y_1,y_2,y_3\leq 1$.
    Let us simplify \eqref{eq:objective} as follows: we first note that in the minimizing solution $y_1 = 0$, as otherwise if one defined $y_1^\prime=0$ and $y_2^\prime=y_2 + \frac{|T|-1}{|T|} \cdot y_1$, $(\alpha^\prime, p^\prime(i), 0, y^\prime_2, y_3)$ would still be feasible and the objective evaluated on this tuple would be strictly smaller than the original one. To show the same with respect to $y_3$, we observe that for $y^\prime_2=\min\{ \frac{k}{2|T|}, y_2 + y_3 \cdot \frac{k}{|T|}\}$, and $y_3^\prime=0$, condition \eqref{eq:constraint} holds since either $y_2^\prime =\frac{k}{2|T|}$, and \eqref{eq:constraint} holds since $\alpha^\prime \geq \frac{1}{10}$ or $y_2^\prime=y_2 + y_3 \cdot \frac{k}{|T|}$ and \eqref{eq:constraint} holds since $y_2^\prime \cdot|T| = y_2 \cdot |T| + y_3 \cdot k$.
    It also holds that $y_2^\prime \leq 1$, and thus, for $y_2^\prime= y_2+y_3 \cdot \frac{k}{|T|} \leq 1$, it holds that $(\alpha^\prime, p^\prime(i), 0, y^\prime_2, 0)$ would still be feasible and the objective evaluated on this tuple would be strictly smaller than the original one.

    Now, we observe that either $p^\prime(x)=0$, which implies that $\alpha^\prime \cdot f(g) -c(g) \geq \alpha^\prime f(x)-c(x)$, thus $\alpha^\prime \geq \frac{9}{70}$, and $u_\principal(g,\alpha^\prime,p^\prime,g) \leq (1-\alpha^\prime)f(g) \leq \frac{61}{70}$ which is strictly less than the utility of $\cI_T$. Otherwise ($p^\prime(x) > 0$), and \eqref{eq:constraint-x} can be replaced with equality, as we could lower the value of $p^\prime(x)$ (up to equality or $0$) and still be feasible. Similarly, once $y_1=y_3=0$, either $y_2=0$ which implies that $\alpha^\prime \cdot f(g)-c(g) \geq \alpha^\prime f(i) -c(i)$, thus $\alpha^\prime \geq \frac{9}{80}$, and $u_\principal(g,\alpha^\prime,p^\prime,g) \leq (1-\alpha^\prime ) f(g) \leq \frac{71}{80}$, which is strictly less than the utility of $\cI_T$. Thus, also \eqref{eq:constraint} can be replaced with equality. Hence, we can write the objective function in \eqref{eq:objective} as
    \begin{align*}
        &\alpha^\prime f(g) + \frac{1}{40} \cdot \left(1 - \frac{\alpha^\prime f(g)-c(g)+c(x)}{\alpha^\prime f(x)}\right) + \frac{1}{80k} \cdot \frac{k}{|T|} \cdot \left(1 - \frac{\alpha^\prime f(g)-c(g)+c(i)}{\alpha^\prime f(i)}\right) \\
        &= \alpha^\prime + \frac{1}{40} \cdot \left(1 - 10\frac{\alpha^\prime -1/10+1/100}{3\alpha^\prime }\right) + \frac{1}{80} \cdot \frac{1}{|T|} \cdot \left(1 - 5\frac{\alpha^\prime-1/10+1/100}{\alpha^\prime}\right)\\
        &= \alpha^\prime + \left( \frac{9}{1200} + \frac{9}{1600\cdot |T|}\right) \cdot \frac{1}{\alpha^\prime} - \frac{7}{120} - \frac{33}{560\cdot |T|}.
    \end{align*}
    This is minimized, for every $|T|$, at $\alpha^\prime = \frac{1}{10}$. In summary, the unique optimal solution is $\alpha^\prime = \frac{1}{10},~p^\prime(x) = \frac{2}{3},~y_1=y_3=0,~y_2=\frac{k}{2|T|}$.
    
    To achieve this minimal cost, the principal must inspect a non-empty set with a probability of exactly $p'(x)$ (otherwise, the lower bound of \eqref{eq:objective} is not tight), which means that all non-empty sets inspected $S$ with $p^\prime(S)>0$ must contain $x$.  
    
    It also must be that each action $i\in [k]$ is inspected with probability $p^\prime(i) \geq \frac{1}{2}$. 
    For every $S\in \cyclic(T)$, let $P_{S}$ be $p^\prime(S \cup \{x\})$.
    To encode $p^\prime(i) \geq \frac{1}{2}$ for all $i \in [k]$, we get the following system of $k$  inequalities in $k$ variables $P_{S}$'s, together with the equality encoding $\sum_{S \in \cyclic(T)} p^\prime(S \cup \{x\}) = \frac{k}{2|T|}$:
    \begin{align*}
        \text{For all } i\in [k],~ \sum_{S\in \cyclic(T): i\in S}P_{S}  &\geq \frac{1}{2}\\
        \sum_{S\in \cyclic(T)} P_{S} &= \frac{k}{2\cdot |T|}.
    \end{align*}
    Since $k$ is a prime number larger than $5$ (thus $0<|T|<k$), the binary vectors representing $P_S$ for $S\in\cyclic(T)$ are independent vectors (over $\R$) thus, the unique solution is where for all $S\in \cyclic(T)$, $P_{S} = \frac{1}{2|T|}$, and therefore $p^\prime(S \cup \{x\}) = \frac{1}{2|T|}$.
    
    To summarize, in the optimal inspection $\cI^\prime_T$, in order to incentivize action $g$, the principal has to inspect sets such that their intersection with the $k$ bad actions belongs to the cyclic shifts of $T$, i.e., $p^\prime(\{x\}) = \frac{2}{3} - \frac{k}{2|T|}$, $p^\prime(\emptyset) = \frac{1}{3}$ and $p^\prime(S \cup \{x\}) = \frac{1}{2|T|}$ for all $S \in \cyclic(T)$. This implies that $\cI^\prime_T = \cI_T$, which is a contradiction.

    $\cI_T$ is, thus, the unique principal's utility maximizing (randomized)  inspection scheme.
\end{proof}

\begin{observation}\label{obs:hard}
    Provided that we know what sets belong to the collection $\cyclic(T)$, we can compute the unique optimal inspection scheme $\cI_T$, inspecting each of the $S \in \cyclic(T)$ with probability exactly $\frac{1}{2|T|}$. The vice versa also holds. Hence, finding any set $S \in \cyclic(T)$ via value or demand queries implies finding the unique optimal inspection scheme $\cI_T$ and vice versa.
\end{observation}

\begin{lemma}\label{lem:queries}
    Given a random hidden set $T \subseteq A$ with $|T| = \left\lceil\frac{4k}{5}\right\rceil$, and  oracle access to function $v_T$, finding any $S \in \cyclic(T)$ requires at least $\left(\frac{5}{4}\right)^k$ value or demand queries to $v_T$.
\end{lemma}
\begin{proof}
    Let us consider any algorithm that uses value and demand queries to find some set  $S \in \cyclic(T)$. We show that it is possible to modify the algorithm only to use value queries (each demand query can be replaced by a single value query), and still find some set $S \in \cyclic(T)$. Therefore, it is sufficient only to bound the number of value queries that are needed to find some set  $S \in \cyclic(T)$.

    First, since $g,\bot$ have an additive value of $1$, we observe that $g,\bot \in D(v_T,q)$ if and only if $q_g \leq 1$ and $q_\bot\leq 1$ respectively, so we can ignore their prices, and find the demand with respect to demand query $q =(q_1,\ldots,q_k,q_x)$.

    Now, we can observe that since the image (after ignoring actions $\bot,g$) of $v_T$ is $$\left\{0,\frac{1}{40},\frac{1}{40}+\frac{1}{80k},\frac{1}{40}+\frac{1}{40k}\right\},$$ and since a demand query returns the least expensive set among sets with the same value, we can partition the subsets of $[k]\cup\{x\}$ into four parts with the same value of $v_T$, and the demand set must be one of the cheapest sets within each part.

    Let us make the following considerations, given a vector of prices $q=(q_1,\ldots,q_k,q_x)$, and letting $\pi:[k]\rightarrow [k]$ be a permutation that satisfies $q_{\pi(1)} \leq \ldots\leq q_{\pi(k)}$: 
    \begin{itemize}
        \item The cheapest set with a value $v_T$ of $0$ is the empty set $\emptyset$ (this is the only set with this value).
        \item The cheapest set with a value $v_T$ of $\frac{1}{40}$ is $\{x\}$ (this is the only set with this value).
        \item The cheapest set with a value $v_T$ of $\frac{1}{40}+\frac{1}{80k}$ is $\{\pi(1)\}$.
        \item If $\{\pi(1),\ldots,\pi(|T|)\} \notin \cyclic(T)$, then the cheapest set with a value $v_T$ of $\frac{1}{40}+\frac{1}{40k}$ is $\{\pi(1),\ldots,\pi(|T|)\}$.
    \end{itemize}
    From above, we have that, if $\{\pi(1),\ldots,\pi(|T|)\} \notin \cyclic(T) $ then in order to calculate the demand $D(v_T,q)$, one only needs to know $v_T(\pi(1),\ldots,\pi(|T|))$. Thus, we can modify the original algorithm as follows: Every time the original algorithm asks a demand query to compute $D(v_T,q)$, the modified algorithm asks the value query for the set $\{\pi(1),\ldots,\pi(|T|)\}$. If $v_T(\{\pi(1),\ldots,\pi(|T|)\}) = \frac{1}{40}+\frac{1}{40k}$, it means that $\{\pi(1),\ldots,\pi(|T|)\} \notin \cyclic(T)$, and so the modified algorithm (as well as the original one) keep on interrogating the value (resp. demand) oracle. Otherwise, we have that $\{\pi(1),\ldots,\pi(|T|)\} \in\cyclic(T)$, and the modified algorithm can stop and return this set.
    
    We next lower bound the number of value queries needed for the modified algorithm to find some set $S \in \cyclic(T)$. Note that, all $k$ sets belonging to $\cyclic(T)$ have size $\left\lceil\frac{4k}{5}\right\rceil$ and the number of subsets of $[k]$ of size $\left\lceil\frac{4k}{5}\right\rceil$ is $k \choose \left\lceil\nicefrac{4k}{5}\right\rceil$. Therefore, each value query on sets $S \notin \cyclic(T)$ only reveals the information that sets $S^\prime \in \cyclic(S)$ do not belong to $\cyclic(T)$. In other words, each value query can exclude at most $k$ sets from the search of $\cyclic(T)$.  The probability of the $t$-th try to find a set in $\cyclic(T)$ is at most $\frac{k}{{k \choose \left\lceil\nicefrac{4k}{5}\right\rceil} - (t-1)\cdot k}$ where the numerator $k$ is the number of sets in $\cyclic(T)$ and the denominator is the remaining candidates after $t-1$ tries. Then, via union bound, this implies that the expected number of value queries needed to find \emph{any} set belonging to $\cyclic(T)$ is
    \begin{align*}
        \sum_{t \in \left[\frac{1}{k} \cdot {k \choose \left\lceil\nicefrac{4k}{5}\right\rceil}\right]} \frac{k}{{k \choose \left\lceil\nicefrac{4k}{5}\right\rceil} - (t-1)\cdot k} \cdot t &\geq \frac{1}{k} \cdot \left({k \choose \left\lceil\nicefrac{4k}{5}\right\rceil} + k\right) \cdot \left(H\left(\frac{{k \choose \left\lceil\nicefrac{4k}{5}\right\rceil}}{k}\right) - 1\right) \geq \frac{1}{k} \cdot \left(\frac{5}{4}\right)^k,
    \end{align*}
    where $H(\ell) = \sum_{r\in[\ell]} \frac{1}{r}$ is the $\ell$-th Harmonic number, and we have used that ${v \choose b} \geq \left(\frac{v}{b}\right)^b$, as well as the fact that $k$ is large enough. This concludes the proof.
\end{proof}

In summary, Lemma~\ref{lem:unique} shows that for the constructed XOS inspection cost function (which is XOS by Lemma~\ref{lem:xos}), the unique optimal inspection scheme has to inspect according to $\cI_T$. Due to Observation~\ref{obs:hard}, we know that computing the optimal inspection scheme is as hard as finding (through value or demand queries) any set in $\cyclic(T)$. Finally, Lemma~\ref{lem:queries} shows that the number of value or demand queries to $v_T$ that is needed for finding a set in $\cyclic(T)$, is exponential in the number of actions $n$. This concludes the proof of Theorem~\ref{thm:impossibility}. \qed

\section{Inapproximability for Subadditive Inspection Cost} \label{sec:subadditive}

In this section, we present a stronger hardness result for the class of subadditive inspection cost functions. In particular, we show that finding a better than $\nicefrac{22}{21}$-approximation requires an exponential number of queries. 
\begin{theorem}\label{thm:inapproximability}
    Every algorithm that for every instance $(A,f,c,\insp) $ with subadditive inspection cost function $\insp$  returns a $\nicefrac{22}{21}$-approximation to the optimal randomized inspection scheme, must use at least an exponential number of value (or demand) queries. 
\end{theorem}

\subsection{Hard Instance Construction}
Let us consider a family of instances $\inst_{X,Y,Z}=(A, f, c, v_{X,Y,Z})$ over the action set $A = \{\bot, g\} \cup [3n]$, with $f(\bot) = 0, f(g) = 1, f(i) = \frac{3}{25} ~\forall~ i \notin \{\bot, g\}$, and  $c(\bot) = 0, c(g) = \frac{3}{4}, c(i) = \frac{3}{100} ~\forall~ i \notin \{\bot, g\}$, parametrized by a balanced partition $(X,Y,Z)$ of $[3n]$, i.e.,  $|X| = |Y| = |Z| = n$ where $X,Y,Z$ are disjoint. Given parameters $X,Y,Z$, the inspection cost function $v_{X,Y,Z}$ satisfies:
\begin{align}\label{eq:subadditive}
    v_{X,Y,Z}(S) \eqdef  \begin{cases} 
    0 & \text{ if } S= \emptyset \\ 
    \frac{3}{100} & \text{ if } g\notin S \text{ and } S\neq \emptyset   \text{ and } (|S\cap[3n]| \leq \frac{3n}{2} \text{ or } S\cap X =\emptyset \text{ or } S\cap Y =\emptyset \text{ or } S\cap Z =\emptyset) \\
        \frac{3}{50} & \text{ if } g\notin S  \text{ and } (|S\cap[3n]| > \frac{3n}{2} \text{ and } S\cap X \neq \emptyset \text{ and } S\cap Y \neq \emptyset \text{ and } S\cap Z \neq \emptyset) \\            1 & \text{ if } g\in S. 
    \end{cases} \nonumber
\end{align}
Our inspection cost function is built such that inspecting $g$ is always suboptimal as it is better to inspect all other actions, and for the other cases, the cost is $3/100$ if the set is small enough (up to size $3n/2$, or if it contained in $2$ out of the three parts of the partition $(X,Y,Z)$.
We first show that for every partition $(X,Y,Z)$, indeed the inspection cost function is subadditive:

\begin{lemma}\label{lem:subadditive}
    For every balanced partition $(X,Y,Z)$ of $[3n]$, it holds that $v_{X,Y,Z}$ is subadditive.
\end{lemma}
\begin{proof}
    For any two sets $S,S^\prime \neq \emptyset$, consider two cases: 
    If  $g \in S \cup S^\prime$, it must belong to at least $S$ or $S'$. Without loss of generality assume that $g\in S$. Then $v_{X,Y,Z}(S) =1 = v_{X,Y,Z}(S\cup S')$, which satisfies the subadditivity condition.

    Otherwise ($g \notin S \cup S^\prime$), since both $S,S'$ are not empty, it holds that $v_{X,Y,Z}(S),v_{X,Y,Z}(S')\geq \frac{3}{100}$, and since $g\notin S\cup S'$, it holds that $v_{X,Y,Z}(S\cup S') \leq \frac{3}{50}$. Overall, we obtain that $v_{X,Y,Z}(S)+v_{X,Y,Z}(S') \geq v_{X,Y,Z}(S\cup S')$, which concludes the proof of the lemma.
    \end{proof}

\subsection{Proof of Theorem~\ref{thm:inapproximability}}

We prove Theorem~\ref{thm:inapproximability} by showing the following lemmas.
We first show in Lemma~\ref{lem:opt-inspect-subadditive} that for any balanced partition $(X,Y,Z)$, there exists a randomized inspection scheme that gives a utility of $0.22$ for the principal.
Then, we show in Lemma~\ref{lem:subadd-non-opt} that any inspection scheme that does not inspect a ``good'' set (a set of size larger than $3n/2$ that belongs to at most two parts of the partition $(X,Y,Z)$) cannot provide a utility of better than $0.21$ to the principal. Lastly, we show that for an unknown random balance partition $(X,Y,Z)$, finding a ``good'' set requires an exponential number of value queries to $v_{X,Y,Z}$. We remark that similarly to our previous construction in Section~\ref{sec:xos}, demand queries to $v_{X,Y,Z}$ can be simulated by value queries, which concludes the proof of the theorem.

\begin{lemma}\label{lem:opt-inspect-subadditive}
    For every balanced partition $(X,Y,Z)$ of $ [3n]$, thee exists a randomized inspection scheme with principal's utility  of  $0.22$.
\end{lemma}

\begin{proof}
Consider  the  inspection scheme
$\cI_{X,Y,Z}=\left(g, \alpha= \frac{3}{4}, p\right)$, where $p([3n] \setminus X) = p([3n] \setminus Y) = p([3n] \setminus Z) = \frac{1}{3}$ and $p(S)=0$ for all other $S$ (i.e., the principal inspects two random parts of the partition $(X,Y,Z)$).
    Under inspection scheme $\cI_{X,Y,Z}$, we have that $p(\bot)=p(g)=0, p(i) = \frac{2}{3}$ for all $i \notin \{\bot, g\}$. Thus, the agent's best response is to take action $g$ since
    \begin{align*}
        u_\agent\left(g, \frac{c(g)}{f(g)}, p, \bot\right) &= f(\bot) \cdot \frac{c(g)}{f(g)} \cdot (1 - p(\bot) - p(g)) - c(\bot) = 0\\
        u_\agent\left(g, \frac{c(g)}{f(g)}, p, i\right) &= f(i) \cdot \frac{c(g)}{f(g)} \cdot (1 - p(i) - p(g)) - c(i) = \frac{3}{100} - \frac{3}{100} = 0, ~\forall~i \notin \{\bot, g\}\\
        u_\agent\left(g, \frac{c(g)}{f(g)}, p, g\right) &= f(g) \cdot \frac{c(g)}{f(g)} - c(g) = 0.
    \end{align*}
    Thus, the agent does not strictly favor any other action over  $g$. Therefore, the principal's utility is
    \begin{align*}
        u_{\principal}(g, \alpha, p,g) &= \left(1 - \frac{c(g)}{f(g)}\right) \cdot f(g) - \sum_{T = X,Y,Z} p([3n] \setminus T) \cdot v_{X,Y,Z}([3n] \setminus T)\\
        &= \frac{1}{4} - \frac{3}{100} = 0.22,
    \end{align*}
    which concludes the proof.
\end{proof}

We next show that any inspection scheme that does not inspect a set of size larger than $3n/2$ that belongs to at most two parts of the partition $(X,Y,Z)$ extract no more than $0.21 $ utility for the principal.

\begin{lemma} \label{lem:subadd-non-opt}
    Any inspection scheme that does no inspect with a positive probability a set $S$ with size $|S|>3n/2$ that is contained in at most two parts of the partition $(X,Y,Z)$, gives a principal's utility of at most  $0.21$.
\end{lemma}
\begin{proof}
    We first note that in order to guarantee a principal's utility of more than $0.21$, the principal must incentivize action $g$ as all other actions have a social welfare (which bounds the principal's utility) of at most $\frac{9}{100}$.
    Moreover, we know that the principal must pay the agent at least the cost of $g$, so the principal must use a contract $\alpha$ such that $\alpha\in [\frac{3}{4},1]$.
    
    Let $p$ be the distribution over inspected sets.
    We can assume that $p(g)=0$ as it is cheaper to inspect $A\setminus\{g\}$ instead of a set that includes $g$ and it reveals the same information. We can also assume that $p(\bot)=0$ since a contract of $\alpha\geq \frac{3}{4}$ guarantees that the agent prefers action $g$ over $\bot$ (so by the monotonicity of the inspection cost function, this action does not need to be inspected). 
    Let $i^\star=\arg\min_{i\in [3n]} p(i)$.
    We know that $\alpha\cdot f(g)-c(g)\geq \alpha\cdot (1-p(i^\star))\cdot  f(i^\star)-c(i^\star)$ which by rearranging we get that implies that 

        \begin{equation}\label{eq:pi}
        p(i^\star) \geq \frac{6}{\alpha} -\frac{22}{3}.
    \end{equation}

    On the other hand, the expected inspection cost satisfies:
    \begin{eqnarray}
        \sum_{S} p(S)\cdot v_{X,Y,Z}(S) & = & \sum_{S: 0< |S\cap [3n]|  \leq 3n/2} p(S) \cdot v_{X,Y,Z}(S) +\sum_{S: |S\cap [3n]| > 3n/2} p(S) \cdot v_{X,Y,Z}(S) \nonumber \\
        & = & \sum_{S: 0<|S\cap [3n]|  \leq 3n/2} p(S) \cdot \frac{3}{100} +\sum_{S: |S\cap [3n]| > 3n/2} p(S) \cdot \frac{3}{50},  \label{eq:cost}
    \end{eqnarray}
    where the second equality is because of the assumption of the lemma that $p$ does not inspect a set of size greater than $3n/2$ that is contained in at most two parts of the partition $(X,Y,Z)$, thus, all sets in the second sum have costs of $\frac{3}{50}$.

    It also holds that 
    \begin{equation}
        3n\cdot p(i^\star) \leq \sum_{i\in [3n]} p(i) = \sum_{S} p(S) \cdot |S\cap [3n] |  \leq \sum_{S:0< |S\cap [3n]|  \leq 3n/2} p(S) \cdot 3n/2 +\sum_{S: |S\cap [3n]| > 3n/2} p(S) \cdot 3n, \nonumber
    \end{equation}
    which by dividing by $3n$ we obtain that 
    \begin{equation} \label{eq:i}
         p(i^\star) \leq  \sum_{S:0< |S\cap [3n]|  \leq 3n/2} p(S)/2 +\sum_{S: |S\cap [3n]| > 3n/2} p(S).
    \end{equation}

    Combining Equations~\eqref{eq:cost} and \eqref{eq:i} we obtain that
    $$  \sum_{S} p(S)\cdot v_{X,Y,Z}(S) \geq \frac{3}{50} \cdot p(i^\star),
    $$
    therefore, the principal's utility is bounded by $$ (1-\alpha)f(g)- \sum_{S} p(S)\cdot v_{X,Y,Z}(S) \leq 1-\alpha -\frac{3}{50} \cdot p(i^\star) \stackrel{\eqref{eq:pi}}{\leq}  1- \alpha -\frac{3}{50} \cdot \left(\frac{6}{\alpha} - \frac{22}{3}\right), $$
    where the last expression is minimized in the interval $[0.75,1]$ at $0.75$, which gives the principal a utility of $0.21$.
\end{proof}

To conclude the proof of Theorem~\ref{thm:inapproximability}, we show that finding a set $S$ of size larger than $3n/2$ that is contained in at most two parts of the partition $(X,Y,Z)$ requires an exponential number of value queries to $v_{X,Y,Z}$.
\begin{lemma}\label{lem:queries-subadd}
    Given a random unknown balanced partition $(X,Y,Z)$ of $[3n]$, an exponential number of value queries to the function $v_{X,Y,Z}$ is required in order to find a set $S$ with size $|S|>3n/2$ that is fully contained in one of the sets $[3n]\setminus X, [3n]\setminus Y, [3n]\setminus Z$.
\end{lemma}
\begin{proof}
We first observe that using a value query $v_{X,Y,Z}(T)$ gives only information if $g\notin T$, and if $|T\cap[3n]| > 3n/2$, and that the only information it reveals is whether $T\cap [3n]$ is a subset of at most two parts of the balanced partition $(X,Y,Z)$. Therefore, it is more beneficial to ask queries for sets of size exactly $3n/2+1$ (assuming that $n$ is even), as if $T$ is contained in at most two parts of $(X,Y,Z)$, so does every subset of it.
The number of ``good'' sets $T$ of size $3n/2+1$ is bounded by $ 3\cdot {2n \choose 3n/2+1}$ while the number of overall sets of size $3n/2+1$ is $ {3n \choose 3n/2+1}$. Thus, in order to find a ``good'' set one needs to go over exponential number of value queries. 
\end{proof}

\begin{remark}
We note that Theorem~\ref{thm:inapproximability} holds even if the algorithm has access to demand queries to the inspection cost function as similarly to the hardness in Section~\ref{sec:xos}, demand queries can be simulated by value queries for the class of functions we use in our hardness construction.
\end{remark}

\section{Extension to the Partially-Observable Inspection Model}\label{sec:partial-model}

In this section, we show how to extend our results to the partially-observable model, where the principal when inspecting the set $S$ only knows whether or not \textit{some} action $j \in S$ has been taken by the agent (rather than knowing for each inspected action whether it was taken or not). That is, let $\bb(S) = \ind{j \in S}$ be the bit of feedback the principal receives. Moreover, when only a single bit of feedback is given to the principal, it holds that the information extracted from $S$ and $A \setminus S$ is exactly the same, since $\bb(S) = 1 - \bb(A \setminus S)$. Thus, it is without of loss of generality to assume that the inspection cost function $v$ is a symmetric set function, i.e., $v(S) = v(A\setminus S)$ for every set $S\subseteq A$ (as any non-symmetric inspection scheme $v:2^A \rightarrow \R_{\geq 0} $ can be replaced by the symmetric inspection scheme $v':2^A\rightarrow \R_{\geq 0}$, where $v'(S)=\min\{v(S),v(A\setminus S)\}$). 
In the paritally-observable model, it is no longer without loss of generality to assume that the inspection cost function $v$ is monotone. Mathematically, requiring that $v$ is also monotone implies that $v$ would need to be a constant function. Moreover, conceptually, checking whether the agent took an action within a class of actions might be less costly than checking whether the agent took a specific action within the class. Thus, we no longer assume the inspection cost function $v$ is monotone. We also assume that inspecting the empty set (or equivalently inspecting $A$), which is equivalent to not inspecting, costs $0$. In other words, we allow the principal to not inspect at no cost.

We focus on the case where the inspection cost function $v$ is submodular. This class of submodular cost functions is very broad, and in particular captures many interesting scenarios, e.g., the scenario studied by the concurrent and independent paper of \citet{FallahJ23}, and the extension to multiple safety levels discussed in the introduction.
In \citep{FallahJ23}, they have two types of actions, namely, safe and unsafe actions. The principal wants to incentivize the agent to take a safe action, and can only make an inspection that costs $\kappa$ that checks whether the taken action was a safe one or not. This inspection function can be captured by the submodular cost function of a minimal cut in a weighted graph, where there is a node for every action and the sum of the weight of the edges between the safe actions and the unsafe action is equal to the inspection cost $\kappa$, while all other edges are sufficiently large such that no reasonable inspection scheme will inspect something different than whether the taken action was safe or unsafe. In any reasonable contract with an inspection scheme, the principal will either inspect the empty set (or equivalently $A$) which has a $0$ inspection cost, or the set of unsafe (or equivalently safe) actions which has a cost of $\kappa$. For the example of $L$ safety levels, consider actions set $ A = A_1, \cup \ldots \cup A_L$, where for every $\ell\in [L]$,  $A_\ell$ is the set of actions of safety level $\ell$.
Inspecting whether the safety level of the taken action is at least $\ell$ is $v_\ell$, where $0=v_1 \leq \ldots \leq v_L$.
This scenario can be captured by a submodular inspection cost function $v$, as a minimum cut in the full graph on vertices $A$ with weights $w(e) = \frac{v_{\ell+1}}{\mid A_{\ell} \times A_{\ell+1} \mid }$ for $e \in A_{\ell} \times A_{\ell+1}$ for $\ell \in [L-1] $,  $w(e) =0 $ for edges of non-consecutive levels, and $w(e)$ is arbitrarily large for edges between actions of the same level. The last constraint ensures that any inspection must only inspect entire safety level sets, and if $\ell$ is the maximum inspected safety level in the cut (assuming without loss of generality that the safety level $L$ is not in the cut), then inspecting the safety levels $[\ell]$ has at most the same cost and reveals at least the same information. Thus, the principal always inspects a prefix of the safety levels, which encompasses the scenario of multiple safety levels.

We note that by our assumption of symmetry of the inspection cost function, we will assume throughout the section that the principal only inspects sets of actions that do not include the suggested action.

\subsection{Deterministic Inspection Schemes}
Our first result is that it is possible to find the best deterministic inspection scheme for this problem in polynomial time for the case of submodular inspection cost.

\begin{theorem}\label{thm:det-partial-positive}
    In the partially-observable model, there exists an algorithm that, for every instance $(A,c,f,v)$ with a submodular inspection cost function, returns the optimal \textit{deterministic} IC inspection-scheme and runs in time polynomial in $n=|A|$, using at most a polynomial number of value queries to $v$.
\end{theorem}
\begin{proof}
 Our algorithm and analysis follows the same steps as in Section~\ref{sec:deterministic} with the following modification. Since we do not assume monotonicity of the inspected cost function, it might be beneficial to inspect a set that contains $S_i$ or $S_{i,j}$ that does not contain $i$. Thus, for each set $S\neq \{i\}$ that our algorithm from Section~\ref{sec:deterministic} considers, we run a submodular minimization algorithm that finds $S^* \in  \arg\min_{S' \subseteq A \setminus \{i\} \mid S \subseteq S'} v(S') $. The last step is equivalent to unconstrained submodular minimization which can be done in polynomial time with a polynomial number of value queries to $v$ \citep{iwata2001combinatorial}.
 Lastly, instead of considering the set $\{i\}$ we consider the set $A \setminus \{i\}$. The best among these inspection schemes is the best deterministic IC inspection scheme. The proof of optimality is identical to the proof of Theorem~\ref{thm:det}. 
\end{proof}

We remark that it in contrast to the case of monotone inspection cost functions, the result of finding the best deterministic IC inspection scheme cannot be extended to general symmetric inspection cost functions.

\begin{proposition}\label{prop:det-partial-negative}
        In the partially-observable model, no algorithm that uses at most a polynomial number of value queries can return the optimal \textit{deterministic} IC inspection-scheme for every instance $(A,c,f,v)$.
\end{proposition}
\begin{proof}
    Consider the following random instance with $n$ actions denoted by $[n]$ where $f(i) = \frac{i}{n}$, and $c(i)= \frac{i^2}{2n^3}$, where the inspection cost function draws a uniformly random non-empty set $T$ that is a strict subset of  $ A \setminus \{n\}$. Then 
    \begin{align*}
    v(S) = \begin{cases}
       0 &\text{ if } S\in\{\emptyset,T,A\setminus T,A\}\\
       1 &\text{ else } 
    \end{cases}.
\end{align*}
Let $j = \max \{ j' \in A \mid  j'\neq n \wedge  j' \notin T\}$ (such $j$ is well defined since $T$ is a strict subset of  $A\setminus \{n\}$). 
The optimal deterministic inspection scheme is then to suggest action $n$, inspect $T$, and offer the contract of $\frac{n+j}{2n^2}$.  To see that this inspection scheme indeed incentivizes action $n$, observe that for every action $j'<n$, either $j'$ is inspected or $j'\leq j$. Thus, the incentive constraint is that 
$$\alpha=\frac{n+j}{2n^2} \geq \frac{n+j'}{2n^2}=  \frac{c(n)-c(j')}{f(n)-f(j')}$$
is satisfied for every $j'$ that is not inspected, thus, $$\alpha f(n) -c(n) = \alpha \left(f(j') + f(n)-f(j')\right) -c(n) \geq \left(\alpha f(j') + c(n)-c(j')\right)-c(n) =   \alpha f(j')-c(j') .$$

To see that this inspection scheme is optimal, we first observe that it is never beneficial to deterministically inspect a set that has non-zero cost as it causes the principal to have no negative utility. It is better to inspect $T$ over the empty set as it reveals more information at the same cost.
The principal's utility under this inspection scheme is $f(n)(1-\frac{n+j}{2n^2}) \geq 1-\frac{1}{n}$ which is higher than the potential utility from any other action (as all other actions hves a value of $f$ of at most $1-\frac{1}{n}$ and require a non-zero payment in order to be incentivized). Thus the optimal contract must incentivize action $n$. Lastly, the principal cannot offer a contract of less than $\frac{n+j}{2n^2}$ as otherwise the agent will prefer to take action $j$.

Lastly, to show that the optimal contract cannot be found using a polynomial number of value queries, observe that in order to decide whether the payment should be $\frac{2n-1}{2n^2}$ the principal needs to go over all sets containing $n-1$ which takes an exponential number queries.
\end{proof}

\subsection{Randomized Inspection Schemes}

Next, we turn to the case of randomized inspection schemes. We show the following theorem:
\begin{theorem}\label{thm:insp-partial}
    In the partially-observable model, there exists an algorithm that, for every instance $(A,c,f,v)$ with a submodular inspection cost function, returns an additive $ \varepsilon$-approximation to the optimal \textit{randomized} IC inspection-scheme and runs in time polynomial in $n=|A|$ and $\frac{1}{\varepsilon}$, using at most a polynomial number of value queries to $v$.
\end{theorem}

\begin{proof}
    
Our proof utilizes the structure derived in Section~\ref{sec:submod} with the following differences. We now can assume that $p(i)=0$, which removes one variable from our optimization. On the other hand, since we no longer have monotonicity, we cannot partition the values of $\alpha$ to a polynomial number of intervals for which we can write a simple optimization problem with a constant number of parameters and optimize it using a closed formula.
However, after fixing the suggested action $i$ we can still write constraints on $\alpha$ given the marginals of $p \in [0,1]^{A\setminus\{i\}}$ as follows:
\begin{align}
    \min_{\alpha, p} \ &\alpha f(i) +  \insp^L(p) \nonumber\\
    \text{s.t. } &\alpha f(i) - c(i) \geq \alpha f(j) \cdot \left(1 - p(j)\right) - c(j), \forall j \neq i \label{eq:insp-partial} \\
    \qquad &\alpha \in \left[\frac{c(i)}{f(i)},1\right], p(j)\in[0,1] \quad \forall j\neq i  \nonumber,
\end{align}
where $v^L $ is the Lov\'asz extension of $v$.

We can assume that $c(i)\neq 0 $ (as for these actions, the optimal way to incentivize them is by offering a contract of $0$ and doing no inspection).

Thus, $\alpha \geq \frac{c(i)}{f(i)}$ is strictly positive, and by dividing the constraints by $\alpha $, rearranging and renaming $\beta=\frac{1}{\alpha}$, we get the equivalent Program:
\begin{align}
    \min_{\beta, p} \ & \frac{f(i)}{\beta} +  \insp^L(p) \nonumber\\
    \text{s.t. } & f(j)-f(i)-f(j) \cdot p(j) + \beta \cdot (c(i)- c(j)) \leq   0 , \quad \forall j \neq i \label{eq:insp-partial2}\\
    \qquad &\beta \in \left[1, \frac{f(i)}{c(i)}\right], \quad  p(j)\in[0,1] \quad \forall j\neq i  \nonumber. 
\end{align}
This is a convex optimization since the  Lov\'asz extension is convex, and since $\frac{f(i)}{\beta}$ is a convex function, and convex functions are closed under summation, and all constraints are linear.
Moreover, if we denote by  $V=\max_{S} v(S)$, then the objective function is $(V+1)$-Lipschitz since $v^L(p\pm x\cdot e_i) \leq v_L(p) + x\cdot V$ while $\frac{f(i)}{\beta\pm x} \leq x$.

Thus, we can find an additive error of $\varepsilon$ in time polynomial in $n,\frac{1}{\varepsilon},\frac{1}{c(i)},V$. We note that the dependency on $\frac{1}{c(i)}$ can be replaced by $\log(\frac{1}{c(i)})$ since as $\beta$ is further than $1$, its impact is reduced.

Solving Program~\eqref{eq:insp-partial2} up to an additive error of $\varepsilon$ for every action $i$ gives the optimal inspection scheme up to an additive loss in the utility of $\varepsilon$.
\end{proof}

We note that in contrast to our result in Section~\ref{sec:submod}, we do not have a closed formula for the optimal contract. Thus, this only shows that our algorithm only converges polynomially to the optimal contract.

\section{Discussion}

We introduce a principal-agent model that relaxes the hidden-action assumption and that allows the principal to incentivize the agent through both positive incentives (payments) and negative incentives (inspections). 
To begin our analysis, we focus on deterministic IC inspection schemes where the principal inspects a subset of actions with certainty. We demonstrate that finding the best deterministic IC inspection scheme can be done efficiently for all (monotone) cost functions.
However, committing to deterministic inspection schemes can lead to significantly reduced utilities. As a result, we explore the use of randomized inspection schemes, which are commonly employed in practice. We then provide an efficient method for determining the optimal randomized IC inspection scheme when the inspection cost function is submodular.  
We complement our positive result by showing that the problem becomes intractable when dealing with XOS inspection cost functions. This highlights the class of submodular functions as the boundary of tractability for this problem.

Although our negative result for XOS functions rules out exact optimization, it still leaves open the possibility of efficiently approximating the optimum—for example, through a PTAS.
We ruled out the existence of a PTAS for the case of subadditive inspection cost functions. However, the question of whether a constant-factor approximation is achievable remains open.

While our work primarily focuses on the principal suggesting a single action, an intriguing direction for further exploration would involve considering scenarios where the principal can suggest a subset of actions. Furthermore, investigating the contract with inspections model in the presence of non-binary outcomes, an agent that can choose subsets of actions, multiple agents working on the same project, or hidden types would provide interesting avenues for future study.

\section*{Acknowledgements}
Tomer Ezra is supported by the Harvard University Center of Mathematical Sciences and Applications.
Stefano Leonardi and Matteo Russo are partially supported by the FAIR (Future Artificial Intelligence Research) project PE0000013, funded by the NextGenerationEU program within the PNRR-PE-AI scheme (M4C2, investment 1.3, line on Artificial Intelligence), and by the MUR PRIN grant 2022EKNE5K (Learning in Markets and Society).

\bibliography{references}

\begin{thebibliography}{38}
\providecommand{\natexlab}[1]{#1}
\providecommand{\url}[1]{\texttt{#1}}
\expandafter\ifx\csname urlstyle\endcsname\relax
  \providecommand{\doi}[1]{doi: #1}\else
  \providecommand{\doi}{doi: \begingroup \urlstyle{rm}\Url}\fi

\bibitem[Alon et~al.(2021{\natexlab{a}})Alon, D{\"{u}}tting, and
  Talgam{-}Cohen]{AlonDT21}
T.~Alon, P.~D{\"{u}}tting, and I.~Talgam{-}Cohen.
\newblock Contracts with private cost per unit-of-effort.
\newblock In \emph{{EC}}, pages 52--69. {ACM}, 2021{\natexlab{a}}.

\bibitem[Alon et~al.(2021{\natexlab{b}})Alon, Lavi, Shamash, and
  Talgam{-}Cohen]{AlonLST21}
T.~Alon, R.~Lavi, E.~S. Shamash, and I.~Talgam{-}Cohen.
\newblock Incomplete information {VCG} contracts for common agency.
\newblock In \emph{{EC}}, page~70. {ACM}, 2021{\natexlab{b}}.

\bibitem[Alon et~al.(2023)Alon, Duetting, Li, and Talgam{-}Cohen]{AlonDLT23}
T.~Alon, P.~Duetting, Y.~Li, and I.~Talgam{-}Cohen.
\newblock Bayesian analysis of linear contracts.
\newblock In \emph{{EC}}, page~66. {ACM}, 2023.

\bibitem[Assadi et~al.(2021)Assadi, Kesselheim, and Singla]{AssadiKS21}
S.~Assadi, T.~Kesselheim, and S.~Singla.
\newblock Improved truthful mechanisms for subadditive combinatorial auctions:
  Breaking the logarithmic barrier.
\newblock In \emph{{SODA}}, pages 653--661. {SIAM}, 2021.

\bibitem[Babaioff et~al.(2012)Babaioff, Feldman, Nisan, and
  Winter]{BabaioffFNW12}
M.~Babaioff, M.~Feldman, N.~Nisan, and E.~Winter.
\newblock Combinatorial agency.
\newblock \emph{J. Econ. Theory}, 147\penalty0 (3):\penalty0 999--1034, 2012.

\bibitem[Ball and Knoepfle(2023)]{BallK23}
I.~Ball and J.~Knoepfle.
\newblock Should the timing of inspections be predictable?
\newblock In \emph{{EC}}, page 206. {ACM}, 2023.

\bibitem[Carroll(2015)]{Carrol15}
G.~Carroll.
\newblock Robustness and linear contracts.
\newblock \emph{American Economic Review}, 105\penalty0 (2):\penalty0 536--63,
  2015.

\bibitem[Castiglioni et~al.(2022{\natexlab{a}})Castiglioni, Marchesi, and
  Gatti]{CastiglioniMG22}
M.~Castiglioni, A.~Marchesi, and N.~Gatti.
\newblock Bayesian agency: Linear versus tractable contracts.
\newblock \emph{Artif. Intell.}, 307:\penalty0 103684, 2022{\natexlab{a}}.

\bibitem[Castiglioni et~al.(2022{\natexlab{b}})Castiglioni, Marchesi, and
  Gatti]{CastiglioniMG22-2}
M.~Castiglioni, A.~Marchesi, and N.~Gatti.
\newblock Designing menus of contracts efficiently: The power of randomization.
\newblock In \emph{{EC}}, pages 705--735. {ACM}, 2022{\natexlab{b}}.

\bibitem[Castiglioni et~al.(2023)Castiglioni, Marchesi, and
  Gatti]{CastiglioniM023}
M.~Castiglioni, A.~Marchesi, and N.~Gatti.
\newblock Multi-agent contract design: How to commission multiple agents with
  individual outcomes.
\newblock In \emph{{EC}}, pages 412--448. {ACM}, 2023.

\bibitem[Correa and Cristi(2023)]{CorreaC23}
J.~R. Correa and A.~Cristi.
\newblock A constant factor prophet inequality for online combinatorial
  auctions.
\newblock In \emph{{STOC}}, pages 686--697. {ACM}, 2023.

\bibitem[Dai and Toikka(2022)]{DaiT22}
T.~Dai and J.~Toikka.
\newblock Robust incentives for teams.
\newblock \emph{Econometrica}, 90\penalty0 (4):\penalty0 1583--1613, 2022.

\bibitem[Deo-Campo~Vuong et~al.(2024)Deo-Campo~Vuong, Dughmi, Patel, and
  Prasad]{Deocampo24}
R.~Deo-Campo~Vuong, S.~Dughmi, N.~Patel, and A.~Prasad.
\newblock On supermodular contracts and dense subgraphs.
\newblock In \emph{Proceedings of the 2024 Annual ACM-SIAM Symposium on
  Discrete Algorithms (SODA)}, pages 109--132. SIAM, 2024.

\bibitem[Dobzinski et~al.(2005)Dobzinski, Nisan, and Schapira]{DobzinskiNS05}
S.~Dobzinski, N.~Nisan, and M.~Schapira.
\newblock Approximation algorithms for combinatorial auctions with
  complement-free bidders.
\newblock In \emph{{STOC}}, pages 610--618. {ACM}, 2005.

\bibitem[D{\"{u}}tting et~al.(2019)D{\"{u}}tting, Roughgarden, and
  Talgam{-}Cohen]{DuttingRT19}
P.~D{\"{u}}tting, T.~Roughgarden, and I.~Talgam{-}Cohen.
\newblock Simple versus optimal contracts.
\newblock In \emph{{EC}}, pages 369--387. {ACM}, 2019.

\bibitem[D{\"{u}}tting et~al.(2021{\natexlab{a}})D{\"{u}}tting, Ezra, Feldman,
  and Kesselheim]{DuttingEFK21}
P.~D{\"{u}}tting, T.~Ezra, M.~Feldman, and T.~Kesselheim.
\newblock Combinatorial contracts.
\newblock In \emph{{FOCS}}, pages 815--826. {IEEE}, 2021{\natexlab{a}}.

\bibitem[D{\"{u}}tting et~al.(2021{\natexlab{b}})D{\"{u}}tting, Roughgarden,
  and Talgam{-}Cohen]{DuttingRT21}
P.~D{\"{u}}tting, T.~Roughgarden, and I.~Talgam{-}Cohen.
\newblock The complexity of contracts.
\newblock \emph{{SIAM} J. Comput.}, 50\penalty0 (1):\penalty0 211--254,
  2021{\natexlab{b}}.

\bibitem[D{\"{u}}tting et~al.(2023)D{\"{u}}tting, Ezra, Feldman, and
  Kesselheim]{DuttingEFK23}
P.~D{\"{u}}tting, T.~Ezra, M.~Feldman, and T.~Kesselheim.
\newblock Multi-agent contracts.
\newblock In \emph{{STOC}}, pages 1311--1324. {ACM}, 2023.

\bibitem[Dutting et~al.(2024)Dutting, Feldman, and Gal~Tzur]{Dutting24}
P.~Dutting, M.~Feldman, and Y.~Gal~Tzur.
\newblock Combinatorial contracts beyond gross substitutes.
\newblock In \emph{Proceedings of the 2024 Annual ACM-SIAM Symposium on
  Discrete Algorithms (SODA)}, pages 92--108. SIAM, 2024.

\bibitem[D{\"{u}}tting et~al.(2024)D{\"{u}}tting, Feldman, and
  Talgam{-}Cohen]{DBLP:journals/fttcs/DuttingFT24}
P.~D{\"{u}}tting, M.~Feldman, and I.~Talgam{-}Cohen.
\newblock Algorithmic contract theory: {A} survey.
\newblock \emph{Found. Trends Theor. Comput. Sci.}, 16\penalty0 (3-4):\penalty0
  211--412, 2024.
\newblock \doi{10.1561/0400000113}.
\newblock URL \url{https://doi.org/10.1561/0400000113}.

\bibitem[Emek and Feldman(2012)]{EmekF12}
Y.~Emek and M.~Feldman.
\newblock Computing optimal contracts in combinatorial agencies.
\newblock \emph{Theor. Comput. Sci.}, 452:\penalty0 56--74, 2012.

\bibitem[Ezra et~al.(2024)Ezra, Feldman, and
  Schlesinger]{ezra_et_al:LIPIcs.ITCS.2024.44}
T.~Ezra, M.~Feldman, and M.~Schlesinger.
\newblock {On the (In)approximability of Combinatorial Contracts}.
\newblock In \emph{15th Innovations in Theoretical Computer Science Conference
  (ITCS 2024)}, volume 287, pages 44:1--44:22. Schloss Dagstuhl --
  Leibniz-Zentrum f{\"u}r Informatik, 2024.

\bibitem[Fallah and Jordan(2024)]{FallahJ23}
A.~Fallah and M.~Jordan.
\newblock Contract design with safety inspections.
\newblock In \emph{Proceedings of the 25th ACM Conference on Economics and
  Computation}, pages 616--638, 2024.

\bibitem[Feldman et~al.(2015)Feldman, Gravin, and Lucier]{FeldmanGL15}
M.~Feldman, N.~Gravin, and B.~Lucier.
\newblock Combinatorial auctions via posted prices.
\newblock In \emph{{SODA}}, pages 123--135. {SIAM}, 2015.

\bibitem[Gottlieb and Moreira(2022)]{GottliebM22}
D.~Gottlieb and H.~Moreira.
\newblock Simple contracts with adverse selection and moral hazard.
\newblock \emph{Theoretical Economics}, 17\penalty0 (3):\penalty0 1357--1401,
  2022.

\bibitem[Grossman and Hart(1983)]{GrossmanH83}
S.~J. Grossman and O.~D. Hart.
\newblock An analysis of the principal-agent problem.
\newblock \emph{Econometrica}, 51\penalty0 (1):\penalty0 7--45, 1983.

\bibitem[Guruganesh et~al.(2021)Guruganesh, Schneider, and
  Wang]{GuruganeshSW21}
G.~Guruganesh, J.~Schneider, and J.~R. Wang.
\newblock Contracts under moral hazard and adverse selection.
\newblock In \emph{{EC}}, pages 563--582. {ACM}, 2021.

\bibitem[Ho et~al.(2016)Ho, Slivkins, and Vaughan]{HoSV16}
C.~Ho, A.~Slivkins, and J.~W. Vaughan.
\newblock Adaptive contract design for crowdsourcing markets: Bandit algorithms
  for repeated principal-agent problems.
\newblock \emph{J. Artif. Intell. Res.}, 55:\penalty0 317--359, 2016.

\bibitem[Holmstrom and Milgrom(1987)]{HolmstromM87}
B.~Holmstrom and P.~Milgrom.
\newblock Aggregation and linearity in the provision of intertemporal
  incentives.
\newblock \emph{Econometrica}, 55\penalty0 (2):\penalty0 303--328, 1987.

\bibitem[Holmstrom and Milgrom(1991)]{HolmstromM91}
B.~Holmstrom and P.~Milgrom.
\newblock Multitask principal-agent analyses: Incentive contracts, asset
  ownership, and job design.
\newblock \emph{Journal of Law, Economics, and Organization}, 7:\penalty0
  24--52, 1991.

\bibitem[Iwata et~al.(2001)Iwata, Fleischer, and
  Fujishige]{iwata2001combinatorial}
S.~Iwata, L.~Fleischer, and S.~Fujishige.
\newblock A combinatorial strongly polynomial algorithm for minimizing
  submodular functions.
\newblock \emph{Journal of the ACM (JACM)}, 48\penalty0 (4):\penalty0 761--777,
  2001.

\bibitem[Kambhampati(2023)]{Kambhampati23}
A.~Kambhampati.
\newblock {Randomization is optimal in the robust principal-agent problem}.
\newblock \emph{Journal of Economic Theory}, 207\penalty0 (C), 2023.

\bibitem[Kelso and Crawford(1982)]{KelsoC82}
A.~S. Kelso and V.~P. Crawford.
\newblock Job matching, coalition formation, and gross substitutes.
\newblock \emph{Econometrica}, 50\penalty0 (6):\penalty0 1483--1504, 1982.

\bibitem[Lavi and Shamash(2022)]{LaviS22}
R.~Lavi and E.~S. Shamash.
\newblock Principal-agent {VCG} contracts.
\newblock \emph{J. Econ. Theory}, 201:\penalty0 105443, 2022.

\bibitem[Lehmann et~al.(2006)Lehmann, Lehmann, and Nisan]{LehmannLN06}
B.~Lehmann, D.~Lehmann, and N.~Nisan.
\newblock Combinatorial auctions with decreasing marginal utilities.
\newblock \emph{Games Econ. Behav.}, 55\penalty0 (2):\penalty0 270--296, 2006.

\bibitem[Walton and Carroll(2022)]{WaltonC22}
D.~Walton and G.~Carroll.
\newblock A general framework for robust contracting models.
\newblock \emph{Econometrica}, 90\penalty0 (5):\penalty0 2129--2159, 2022.

\bibitem[Yu and Kong(2020)]{YuK20}
Y.~Yu and X.~Kong.
\newblock Robust contract designs: Linear contracts and moral hazard.
\newblock \emph{Oper. Res.}, 68\penalty0 (5):\penalty0 1457--1473, 2020.

\bibitem[Zhu et~al.(2023)Zhu, Bates, Yang, Wang, Jiao, and Jordan]{ZhuBYWJJ23}
B.~Zhu, S.~Bates, Z.~Yang, Y.~Wang, J.~Jiao, and M.~I. Jordan.
\newblock The sample complexity of online contract design.
\newblock In \emph{{EC}}, page 1188. {ACM}, 2023.

\end{thebibliography}
\appendix
\section{Omitted Content}\label{app:omitted}

\subsection{Proof of Theorem~\ref{thm:det}}

\thmdet*

\begin{proof}
The proof proceeds as follows.

\paragraph{Incentive compatibility.} We first show that every inspection scheme considered by our algorithm is indeed IC.
For the inspection scheme defined by $(i^*,0,\emptyset)$, this is trivially true, since the agent's utility for action $i^*$ is $0$, while for every other action cannot be more than $0$.

Now, consider inspection scheme $(i,\frac{c(i)-c(j)}{f(i)-f(j)},S_{i,j})$ where $j\in A_i$. Since $j\in A_i$, then the utility of the agent from selecting  action $i$ is strictly positive, as $\frac{c(i)-c(j)}{f(i)-f(j)} \cdot f(i) -c(i) > \frac{c(i)}{f(i)} \cdot f(i) -c(i) =0$.
On the other hand, for $j^\prime\neq i$, if $j^\prime\in S_{i,j}$ then the agent's utility from selecting $j^\prime$ is at most $0$, since the principal inspects action $j^\prime$.
Else ($j^\prime\notin S_{i,j}$), if $j^\prime\in A_i$, then we get that $j^\prime$ must satisfy that $\frac{c(i)-c(j^\prime)}{f(i)-f(j^\prime)  } \leq  \frac{c(i)-c(j)}{f(i)-f(j)  } $, which implies that, if we denote by $\delta_{i,j,j^\prime} = \frac{c(i)-c(j)}{f(i)-f(j) }-\frac{c(i)-c(j^\prime)}{f(i)-f(j^\prime) } \geq 0 $, then \begin{eqnarray}
	f(i) \cdot \frac{c(i)-c(j)}{f(i)-f(j) } &-& c(i)  =  f(i) \cdot \frac{c(i)-c(j^\prime)}{f(i)-f(j^\prime)} -c(i)   + f(i)\cdot\delta_{i,j,j^\prime} \nonumber\\ & \geq & (f(i)-f(j^\prime)+f(j^\prime)) \cdot \frac{c(i)-c(j^\prime)}{f(i)-f(j^\prime)} -(c(i)-c(j^\prime)+c(j^\prime)) + f(j^\prime)\cdot\delta_{i,j,j^\prime} 
	\nonumber\\ & = &  f(j^\prime)\cdot  \frac{c(i)-c(j^\prime)}{f(i)-f(j^\prime)} -c(j^\prime)+ f(j^\prime)\cdot\delta_{i,j,j^\prime}
	\nonumber\\ & = & f(j^\prime) \cdot \frac{c(i)-c(j)}{f(i)-f(j) }-c(j^\prime) , \label{eq:j}
\end{eqnarray}
where the first and last equalities are by the definition of $\delta_{i,j,j^\prime}$, and  the inequality is since $f(i) >f(j^\prime)$. This implies that the agent does not want to deviate to $j^\prime$.
For $j^\prime \notin A_i$, then if $f(j^\prime)< f(i)$ then the same inequalities as \eqref{eq:j} work since $ \delta_{i,j,j^\prime} =   \frac{c(i)-c(j)}{f(i)-f(j)} - \frac{c(i)}{f(i)} + \frac{c(i)}{f(i)}- \frac{c(i)-c(j^\prime)}{f(i)-f(j^\prime)} \geq  0$.

We are left to consider the case that $f(j^\prime) \geq f(i)$. If $f(j^\prime) = f(i)$, then the condition $j^\prime\notin S_{i,j}$ implies that $c(j^\prime) \geq  c(i)$,  thus the agent cannot strictly favor action $j^\prime$ over $i$. If $f(j^\prime) > f(i)$ then the condition $j^\prime\notin S_{i,j}$ implies that $ \frac{c(i)-c(j)}{f(i)-f(j)} \leq \frac{c(i)-c(j^\prime)}{f(i)-f(j^\prime)}$, which allows us to use Inequality~\eqref{eq:j} by observing that $\delta_{i,j,j^\prime} \leq 0$ and $f(i) \leq f(j^\prime)$.

An equivalent analysis (as for the former case) shows that  the inspection scheme $(i,\frac{c(i)}{f(i)},S_i)$ is IC.

For inspection scheme $\left(i,\frac{c(i)}{f(i)},\{i\}\right)$, the utility from action $i$ is $0$, and from any other action is non-positive since action $i$ is inspected.

\paragraph{Optimality.}
We next show that for every IC inspection scheme $(i,\alpha,S)$, the algorithm considers an inspection scheme with at least as much expected utility for the principal.
The case of $c(i)=0$ is solved trivially by the first inspection scheme that the algorithm considers (inspection scheme $(i^*,0,\emptyset)$).

Else ($c(i) >0$), then it must be that $\alpha f(i) \geq c(i)$ (as otherwise it contradicts the individual rationality of the agent), thus $f(i)>0$, and $\alpha \geq \frac{c(i)}{f(i)}$.
We next consider three cases:
If $i\in S$, then the inspection scheme  $\left(i,\frac{c(i)}{f(i)},\{i\}\right)$ achieves at least the same utility for the principal.
Otherwise 
 ($i\notin S$), if $A_i \subseteq S$, then it must be that $S_i \subseteq S$. Assume towards contradiction that there exists $j^\prime\in S_i\setminus S$, then it must belong to $S_i\setminus A_i$ (since $S$ contains $A_i$). If $f(j^\prime)=f(i)$ then the condition is equivalent to $c(j^\prime)<c(i)$ which means that  the agent can benefit from switching to $j^\prime$ as it has the same expected payment from the principal, but with a lower cost (which is a contradiction to IC). If $f(j^\prime)> f(i)$ then by definition of  $S_i\setminus A_i$, it holds that $\frac{c(i)}{f(i)} > \frac{c(j^\prime)}{f(j^\prime)}$ thus 
 \begin{align*}
     \alpha f(j^\prime) -c(j^\prime) &= \left(\alpha-\frac{c(i)}{f(i)} + \frac{c(i)}{f(i)}\right) \cdot f(j^\prime) -c(j^\prime)  \\
     &> \left(\alpha-\frac{c(i)}{f(i)}\right) \cdot f(j^\prime) \geq \left(\alpha-\frac{c(i)}{f(i)}\right) \cdot f(i) = \alpha f(i) -c(i),
 \end{align*}
 which is a contradiction to IC, thus, inspection scheme $(i,\frac{c(i)}{f(i)},S_{i})$ gives at least the same utility to the principal.

 Else ($\exists j \in  A_i \setminus S$), then consider  $j^*= \arg\max_{j\in  A_i \setminus S} \frac{c(i)-c(j)}{f(i)-f(j)}$.
 It must be that $\alpha \geq \frac{c(i)-c(j^*)}{f(i)-f(j^*)}$, as otherwise
 \begin{eqnarray*}
 	 \alpha f(j^*) -c(j^*)  & = & \left(\alpha -\frac{c(i)-c(j^*)}{f(i)-f(j^*)}+\frac{c(i)-c(j^*)}{f(i)-f(j^*)}\right) \cdot (f(j^*)-f(i)+f(i)) -c(j^*) \\
 	 & = & \left(\alpha -\frac{c(i)-c(j^*)}{f(i)-f(j^*)}\right) \cdot f(j^*)+\frac{c(i)-c(j^*)}{f(i)-f(j^*)} \cdot f(i) -c(i) \\ & > & \left(\alpha -\frac{c(i)-c(j^*)}{f(i)-f(j^*)}\right) \cdot f(i)+\frac{c(i)-c(j^*)}{f(i)-f(j^*)} \cdot f(i) -c(i)
 	 \\ &= & \alpha f(i)-c(i),
 \end{eqnarray*}   
where the inequality is since $\alpha <\frac{c(i)-c(j^*)}{f(i)-f(j^*)}$, and since $f(j^*)<f(i)$.

By definitions of $j^*,S_{i,j^*}$, it holds that $A_i \cap S_{i,j^*} \subseteq S $.
It also holds that since $\alpha \geq \frac{c(i)-c(j^*)}{f(i)-f(j^*)}$, every action $j'\in S_{i,j^*} \setminus A_i $ must be in $S$. Assume towards contradiction that there exists $j' \in  S_{i,j^*} \setminus A_i \setminus S $. Then, $f(j') \geq f(i)$, therefore, it holds that 
\begin{eqnarray*}
    \alpha f(j') -c(j') & \geq  & \left(\alpha-\frac{c(i)-c(j^*)}{f(i)-f(j^*)}\right)   f(i)  + \frac{c(i)-c(j^*)}{f(i)-f(j^*)} f(j') -c(j')
    \\ & > & \left(\alpha-\frac{c(i)-c(j^*)}{f(i)-f(j^*)}\right)   f(i)  + \frac{c(i)-c(j^*)}{f(i)-f(j^*)} f(i) -c(i)
    \\ & = &  \alpha f(i) -c(i),
\end{eqnarray*}
where the first inequality is since $\alpha \geq \frac{c(i)-c(j^*)}{f(i)-f(j^*)}$, and since $f(j') \geq f(i)$, and the second inequality is since  $j' \in S_{i,j^*} \setminus A_i$.
Therefore, $S_{i,j^*} \subseteq S $, thus inspection scheme $(i,\frac{c(i)-c(j^*)}{f(i)-f(j^*)},S_{i,j^*})$ gives at least the same utility to the principal,
which concludes the proof.

\paragraph{Complexity.} The algorithm uses at most $O(n^2)$ value queries since it asks for each $i$ at most $n+1$ queries.
The algorithm runs in polynomial time by design.
\end{proof}

\subsection{Proof of Lemma~\ref{lem:submod}}

\lemsubmod*

\begin{proof}
The first claim holds since $$\sum_{S \subseteq A}  p^\prime(S) = p^\prime(\emptyset) + \sum_{t\in [n]} p^\prime(\{\pi(t),...,\pi(n)\}) =p^\prime(\emptyset) + \sum_{t\in [n]} p(\pi(t)) -p(\pi(t-1))  =   \sum_{S \subseteq A} p(S),$$
where the last equality holds by telescoping and by definition of $p^\prime(\emptyset)$.

    To show the second claim, let $\ell$ be the preimage of $j$ according to bijection $\pi$ , i.e., $j = \pi(\ell)$. We get
    \begin{align*}
        p^\prime(j) &= \sum_{S \subseteq A: j \in S} p^\prime(S) = \sum_{t \leq \ell} p^\prime(\{\pi(t),...,\pi(n)\}) \\
        &= \sum_{t \leq \ell} \left(p(\pi(t)) - p(\pi(t-1)) \right)= p(\pi(\ell)) - p(\pi(0)) = p(\pi(\ell)) \\
        &= p(j).
    \end{align*}
    For the third claim, we have
    \begin{align*}
        \sum_{S \subseteq A} p(S) \cdot v(S) &= \sum_{S \subseteq A} p(S) \cdot \left(\sum_{\ell=1}^{n} \ind{\pi(\ell) \in S} \cdot v(\pi(\ell) \mid S \setminus \{\pi(1),\ldots,\pi(\ell)\})\right) \\
        &\geq \sum_{S \subseteq A} p(S) \cdot \left(\sum_{\ell=1}^{n} \ind{\pi(\ell) \in S} \cdot v(\pi(\ell) \mid A \setminus \{\pi(1),\ldots,\pi(\ell)\})\right)\\
        &= \sum_{\ell=1}^{n} v(\pi(\ell) \mid A \setminus \{\pi(1),\ldots,\pi(\ell)\}) \cdot \left(\sum_{S \subseteq A: \pi(\ell) \in S} p(S)\right)\\
        &= \sum_{\ell=1}^{n} p(\pi(\ell)) \cdot v(\pi(\ell) \mid A \setminus \{\pi(1),\ldots,\pi(\ell)\})\\
        &= \sum_{\ell=1}^{n} (p(\pi(\ell)) - p(\pi(\ell - 1)) \cdot v(A \setminus \{\pi(1),\ldots,\pi(\ell-1)\})\\
        &= \sum_{\ell=1}^{n} p^\prime(A \setminus \{\pi(1),\ldots,\pi(\ell-1)\}) \cdot v(A \setminus \{\pi(1),\ldots,\pi(\ell-1)\})\\
        &= \sum_{S \subseteq A} p^\prime(S) \cdot v(S),
    \end{align*}
    where the first equality is by definition of marginals, the inequality is by submodularity of $v$, the second equality is by switching the order of summations, the third equality is by definition of $p(\pi(\ell))$, the fourth equality follows from telescoping, the fifth equality holds by definition of $p^\prime$, and the last equality follows since all added sets to the summation have either $p^\prime(S)=0$ or $S= \emptyset$ for which $v(S)=0$.
\end{proof}

\section{Non Incentive-Compatible Inspection Schemes}
\label{app:non-ic}

In the main body of the paper, we focused on incentive-compatible inspection schemes. Indeed, we observe that the principal may propose an inspection scheme $(j, \alpha, p)$  so as to cause the agent to best respond by some other action $i \neq j$: We call these types of inspection schemes \emph{non-incentive-compatible} (non-IC). 

We show that, for deterministic inspection, IC inspection schemes are without loss of generality in that they can extract at least as much utility as non-IC inspection schemes.

\begin{theorem}
    For every instance $(A,c,f,v)$, and every deterministic non-incentive-compatible inspection scheme $(j, \alpha, S)$ with agent's best response $i$, there exists a deterministic incentive-compatible inspection scheme $(i', \alpha', S')$ with agent's best response $i'$ yielding at least as much utility for the principal.
\end{theorem}
\begin{proof}
    Let us consider a non-incentive-compatible inspection scheme $(j, \alpha, S)$ inducing the agent to take action $i \neq j$ as the best response.
    
    If $\{i,j\} \cap S \neq \emptyset$ (i.e., at least one of $i,j$ is inspected), then we know that the principal pays the agent $0$, which means that her incurred cost $c(i)$ is $0$. Therefore, the principal could simply suggest action $i$, offer $\alpha=0$, and inspect set $S$, and the agent's best response would still be $i$. This means that the deterministic incentive-compatible inspection scheme $(i, 0, S)$ still induces the agent's best response to be $i$, and thus achieves the same utility as the non-incentive-compatible one.
    
    Otherwise, $\{i,j\} \cap S = \emptyset$, and we have that $u_\agent(j,\alpha ,S,\ell) = u_\agent(i,\alpha ,S,\ell)$ for every $\ell \in A$. Thus $i$ is also the best response to  $(i,\alpha,S)$.
\end{proof}
 
For randomized inspection schemes, the same does not hold, as the next example demonstrates.

\begin{example}\label{ex:gap-ic}
    Let us consider the following instance with three actions $A = \{\bot, 1, 2\}$: $c(\bot)=0,~f(\bot)=0,~c(1)=\frac{1}{10},~f(1)=\frac{2}{5},~c(2)=\frac{1}{2},~f(2)=1$, and the additive inspection cost function  $v(S) \eqdef \frac{3}{10} \cdot \ind{1\in S} + 2 \cdot \ind{2\in S}  $.
    \begin{itemize}
        \item Under no inspection, the best the principal can do is incentivizing action $2$ by offering $\alpha=\frac{2}{3}$, which yields a utility of $\frac{1}{3}$. Note that the utility from incentivizing action $1$ is $\frac{3}{10} < \frac{1}{3}$.
        \item The best randomized incentive-compatible inspection scheme for the principal is to suggest action $2$ and solve the following program:
        \begin{align*}
            \min_{\alpha, p(1), p(2)} &\alpha + p(1) \cdot \frac{3}{10} + p(2) \cdot 2\\
            \emph{s.t. }  &\alpha - \frac{1}{2} \geq \alpha \cdot \frac{2}{5} \cdot (1-p(1)-p(2)) - \frac{1}{10}\\
            &\alpha \geq \frac{1}{2}\\
            &0\leq p(1), p(2) \leq 1.
        \end{align*}
        Hence, $\alpha \approx 0.548$, $p(1) \approx 0.326$,  and $p(2)=0$, yielding a principal's utility of approximately $0.355$. This is the best IC-inspection scheme since an IC-inspection scheme that incentivizes action $\bot$ (respectively, $1$) cannot give the principal a utility that is better than $f(\bot)-c(\bot)=0$ (respectively, $f(1)-c(1)=\frac{3}{10}$).
        \item Let us consider the randomized non-incentive-compatible inspection scheme $(\bot, 1, p)$, where $p(\{\bot\})=\frac{1}{2}$, $p(\{1\}) =\frac{1}{4}$, and $p(\emptyset)=\frac{1}{4}$. Then, the agent's utility from taking each of the three actions is:
        \begin{align*}
            u_\agent(\bot, 1, p, \bot) &= 0,\\
            u_\agent(\bot, 1, p, 1) &= \frac{2}{5} \cdot \frac{1}{4} - \frac{1}{10} = 0,\\
            u_\agent(\bot, 1, p, 2) &= \frac{1}{2} - \frac{1}{2} = 0.
        \end{align*}
        From above, we see that action $2$ is in the agent's best response set with respect to inspection scheme $(\bot, 1, p)$. The principal pays $1$ to the agent only with probability $\frac{1}{2}$ (the probability the agent is not caught taking an action different from $\bot$), and incurs an inspection cost of $\frac{3}{10}$ (for inspecting action $1$) only with probability $\frac{1}{4}$. All in all, the principal has utility $\frac{1}{2} - \frac{3}{40} = 0.425 > 0.355$.
        If the principal pays $1+\varepsilon$ instead of $1$, then the unique best response for the agent is to take action $2$.
    \end{itemize}
\end{example}

\section{Gaps between Randomized and Deterministic Inspection Schemes}
\label{app:gap}

In spite of their convenient structural properties, such as computational tractability in finding them, deterministic inspection schemes can be much worse than randomized ones. In particular, we extend Example \ref{ex:gap} and show that the principal's utility gap between deterministic and randomized inspection schemes can grow as the instance size grows. Formally, we prove the following result:
\begin{proposition}\label{prop:gap}
    There exists an instance $(A, f, c, v)$, and an IC randomized inspection scheme $(i, \alpha, p)$ such that for every deterministic IC inspection scheme $(i^\prime,\alpha^\prime,S)$, it holds that 
    \begin{align*}
         u_\principal(i^\prime,\alpha^\prime,S,i^\prime) &\leq \frac{4}{n} \cdot u_\principal(i, \alpha, p, i).
    \end{align*}
\end{proposition}
\begin{proof}
    We present the following instance (Table \ref{tab:gap}) where $v:2^A \rightarrow \R_{\geq 0}$ is $v(S) = |S| \cdot \frac{n}{2^n}$, and $f,c$ are defined in the following table: 
    \begin{table}[H]
        \centering
        \begin{tabular}{|c|c|c|}
         \hline
         $A$ & $f$ & $c$ \\ \hline
         $\bot$ & $0$ & $0$   \\ 
         $1$ & $\frac{4}{2^n}$ & $\frac{2}{2^n}$  \\ 
         $\vdots$ & $\vdots$ & $\vdots$   \\ 
         $i$ & $\frac{2^{i+1}}{2^n}$ & $\frac{2^{i+1}-i-1}{2^n}$   \\ 
         $\vdots$ & $\vdots$ & $\vdots$   \\ 
         $n-1$ & $1$ & $\frac{2^n-n}{2^n}$   \\ \hline
        \end{tabular}
        \caption{\label{tab:gap} Instance exhibiting $\Omega(n)$ gap between deterministic and randomized inspection schemes. 
        Similar instances have been used to study various other types of gaps between optimal contract vs optimal linear contract (in different settings) for example, in \cite{DuttingRT19, GuruganeshSW21}.}
    \end{table}
    We first show that the best deterministic inspection scheme never inspects (i.e., always inspects the empty set). Indeed, any deterministic inspection scheme that inspects a non-empty set must pay at least $\frac{n}{2^n}$, since $v(S) \geq  \frac{n}{2^n}$ for all $S \neq \emptyset$. Since the maximal social welfare arising from the instance (when ignoring incentive constraints) is in correspondence of action $n-1$ and equals $\frac{n}{2^n}$, then the highest expected reward minus payment to an agent a deterministic inspection can extract from the instance is bounded above by the social welfare. This means that the utility by inspecting some non-empty set is at most $0$.

    Hence, since the principal does not inspect, we have, by the incentive-compatibility constraints, that if the principal wants to incentivize the agent to take action $j \neq \bot $, then it must hold that: $$\alpha \cdot \frac{2^{j+1}}{2^n} - \frac{2^{j+1} - j -1}{2^n} \geq \alpha \cdot \frac{2^j}{2^n} - \frac{2^j - j}{2^n},$$ which by rearranging gives us that  $\alpha \geq 1 - \frac{1}{2^j}$. Thus, the utility from incentivizing action $j$ is  at most $(1-\alpha)f(j) = \frac{1}{2^j} \cdot \frac{2^{j+1}}{2^n} = \frac{2}{2^n}$. 
    
    On the other hand, let us consider the randomized IC inspection scheme  with suggested action $n-1$, with contract $\alpha = 1 - \frac{n}{2^n}$, and $p(\emptyset) = \frac{1}{2},~p(\{n-1\}) = \frac{1}{2}$ as our randomized inspection distribution. The utility of the agent for action $n-1$ is $\alpha \cdot f(n-1) -c(n-1) =0$, while for every other action $j\neq \bot,n-1$, we have that
    \begin{align*}
         u_\agent(n-1,\alpha,p,j) = \alpha f(j) \cdot \left(1 - p(n-1) - p(j)\right) - c(j) = \left(1-\frac{n}{2^n}\right) \cdot   \frac{2^{j+1}} {2^n} \cdot \frac{1}{2}  - \frac{2^j-j}{2^n} \leq  0 
    \end{align*}
    Hence, the agent is incentivized to take action $n-1$. Thus, the principal gets
    \[
        u_\principal\left(n-1, 1 - \frac{n}{2^n}, \left(p(\emptyset) = \frac{1}{2},~p(\{n-1\}) = \frac{1}{2}\right), n-1\right) = \frac{n}{2^n} \cdot 1 - \frac{1}{2} \cdot \frac{n}{2^n} = \frac{n}{2^{n+1}},
    \]
    which concludes the proof.
\end{proof}

\begin{remark}
    The instance presented in Table \ref{tab:gap} can be easily modified to have $v(S)=0$ for all $S \subseteq A$, implying that the gap in principal's utility between no inspection and best deterministic inspection can also be at least $\Omega(n)$.  This gap is tight, i.e., it is also at most $O(n)$. From \citet{DuttingRT19}, we know that the best (non-inspecting) contract can guarantee a $\frac{1}{n}$ fraction of the maximum possible welfare of the instance. Since the latter is an upper bound on what the principal can achieve using a randomized inspection scheme, then the best contract that does not use inspection (which is a special case of deterministic inspection schemes) guarantees a $\frac{1}{n}$ fraction of the utility of the best randomized inspection scheme.
\end{remark}

\end{document}